\documentclass{tlp}

\usepackage{stmaryrd}
\usepackage{aopmath}
\usepackage{amssymb}
\usepackage{pstricks}
\usepackage{pst-node}
\usepackage{graphicx}
\usepackage{pifont}

\newtheorem{definition}{Definition}

\newcommand{\ignore}[1]{}

\newcommand{\Bi}{\ensuremath{\mathbf{2}}}
\newcommand{\Tri}{\ensuremath{\mathbf{3}}}
\newcommand{\Quad}{\ensuremath{\mathbf{4}}}
\newcommand{\ident}[1]{\mbox{\textit{#1}}}
\newcommand{\pset}[1]{{\cal P}(#1)}
\newcommand{\fun}{\rightarrow}

\newcommand{\impl}{\mathbin{\Rightarrow}}
\newcommand{\biim}{\mathbin{\Leftrightarrow}}
\newcommand{\nat}{\mathbb{N}}

\newcommand{\Her}{\mathcal{G}}
\newcommand{\Interp}{\mathcal{I}}
\newcommand{\lfp}{\mathbf{lfp}}

\newcommand{\chk}{\checkmark}  

\newcommand{\sset}[2]{\left\{~#1 \left|
      \begin{array}{l}#2\end{array}
    \right.     \right\}}

\newcommand{\tuple}[1]{\ensuremath{\langle #1 \rangle}}

\def\defemb#1#2{\expandafter\def\csname #1\endcsname
  {\relax\ifmmode #2\else\hbox{$#2$}\fi}}
\defemb{cH}{{\cal A}}
\newcommand{\ri}{<\!\!\!<}

\begin{document}

\title[Truth versus Information]%
      {Truth versus Information \\ in Logic Programming}

\author[L. Naish and H. S{\o}ndergaard]{
Lee Naish and Harald S{\o}ndergaard\\
Department of Computing and Information Systems \\
The University of Melbourne, Victoria 3010,
Australia\\
\email{\{lee,harald\}@unimelb.edu.au}}

\submitted{15 June 2012}
\revised{6 February 2013}
\accepted{20 April 2013}

\maketitle

\newpsobject{showgrid}{psgrid}{subgriddiv=1,griddots=10,gridlabels=6pt}

\begin{abstract}

The semantics of logic programs was originally described in terms
of two-valued logic.
Soon, however, it was realised that three-valued logic had some
natural advantages, as it provides distinct values not
only for truth and falsehood, but also for ``undefined''.
The three-valued semantics proposed by Fitting and by Kunen are
closely related to what is computed by a logic program, the third truth
value being associated with non-termination.  
A different three-valued semantics, proposed by Naish, shared much 
with those of Fitting and Kunen but incorporated allowances for
programmer intent, the third truth value being associated with 
underspecification.
Naish used an (apparently) novel ``arrow'' operator to relate the intended
meaning of left and right sides of predicate definitions.  In this paper
we suggest that the additional truth values of Fitting/Kunen and 
Naish are best viewed as duals.  
We use Belnap's four-valued logic, also used elsewhere by Fitting, to unify
the two three-valued approaches.
The truth values are arranged in a bilattice which
supports the classical ordering on truth values as well as the
``information ordering''.
We note that the ``arrow'' operator of Naish
(and our four-valued extension) is essentially the information ordering,
whereas the classical arrow denotes the truth ordering.
This allows us to shed new light on many aspects of logic programming,
including program analysis, type and mode systems, declarative debugging and
the relationships between specifications and
programs, and successive executions states of a program.
This paper is to appear
in Theory and Practice of Logic Programming (TPLP).

\end{abstract}


\section{Introduction}

Logic programming is an important paradigm.  Computers can be seen as
machines which manipulate meaningful symbols and the branch of mathematics
which is most aligned with manipulating meaningful symbols is logic.
This paper is part of a long line of research on what are good choices of
logic to use with a ``pure'' subset of the Prolog programming language.
We ignore the ``non-logical'' aspects of Prolog such as cut and built-ins
which can produce side-effects, and assume a sound form of negation
(ensuring in some way that negated literals are always ground before
being called).

There are several ways in which having a well-defined semantics for
programs is helpful.  First, it can be helpful for implementing a
language (writing a compiler, for example)---it forms a specification
for answering ``what should this program compute''.  Second, it can be
helpful for writing program analysis and transformation tools.  Third, it
can be helpful for verification and debugging---it can allow application
programmers to answer ``does this program compute what I intend'' and,
when the answer is negative, ``why not''.  There is typically imprecision
involved in all three cases.
\begin{enumerate}
\item 
Many languages allow some latitude to the implementor in ways
that affect observable behaviour of the program, for example by not
specifying the order of sub-expression evaluation (C is an example).
Even in pure Prolog, typical
approaches to semantics do not precisely deal with infinite loops
and/or ``floundering'' (when a negative literal never becomes ground).
Such imprecision is not necessarily a good thing, but there is often a
trade-off between precision and simplicity of the semantics.  
\item 
Program analysis tools must provide imprecise information in general if
they are guaranteed to terminate, since the properties they seek to
establish are almost always undecidable.  
\item
Programmers are often only interested
in how their code behaves for some class of inputs. For other inputs
they either do not know or do not care 
(this is in addition to the first point).  
Moreover, it is often convenient for programmers to reason about 
\emph{partial} correctness, setting aside the issue of termination.
\end{enumerate}
A primary aim of this paper is to reconcile two different uses of
many-valued logic for understanding logic programs.
The first use is for the provision of semantic definition,
with the purpose of answering ``what should this program compute?''
The other use is in connection with program specification and
debugging, concerned with answering ``does this program compute what
I intend'' and similar questions involving programmer intent.

A second aim is to show the versatility of four-valued logic
in a logic programming context.
Four-valued logic has been recommended by
\citeN{fitting:JLP_1991},
but primarily as a programming language feature, for
distributed programming.
In that context,
the fourth truth value represents conflicting information
derived from different nodes in a network.
We complement that work by pointing out that motivation for
four-valued logic comes from many other sources, 
even when we restrict attention to sequential programming.
Central to our use of this logic is its support for the ``information
ordering'' as well as the classical ordering on truth values.
Our contributions are:
\begin{itemize}
\item 
We show how Belnap's four-valued logic enables a clean
distinction between a formula/query which is undefined,
or non-denoting, and one which is irrelevant, or inadmissible.
\item 
We use this logic to provide a denotational semantics for 
logic programs which is designed to help a programmer reason 
about partial correctness in a natural way.
This aim is different to the semanticist's traditional objective of 
reflecting runtime behaviour, or aligning denotational
and operational semantics.
\item 
The approximative nature of logic program analysis
naturally fits with the information ordering and we show how
semantic approximation can be expressed in terms of
four truth values.
\item 
We show how four-valued logic helps modelling the concept of 
modes in a moded logic programming language such as Mercury.
\item 
We argue that a four-valued semantics and the information ordering
clarify the relation between programs and formal specifications.
\item
We show how established practice in declarative debugging can be
extended with four values.
\item
Finally, we argue that the computation model of logic programming can
be viewed from the perspective of the information order rather than the
classical truth order.
\end{itemize}
This paper is an extended version of \citeN{Nai-Son-Hor:CATS12}.
We assume the reader has a basic understanding of pure logic
programs, including programs in which clause bodies use negation,
and their semantics.
We also assume the reader has some familiarity with the concepts of
types and modes as they are used in logic programming.

The paper is structured as follows.  
We set the scene in Section~\ref{sec-scene} by revisiting the problems
that surround approaches to logical semantics for pure Prolog.
In Section~\ref{sec-interpretations}
we introduce the three- and four-valued logics and many-valued 
interpretations that the rest of the paper builds upon.
In Section~\ref{sec-lp-semantics} we provide some background on 
different approaches to the semantics of pure Prolog, focusing on
work by Fitting and Kunen.
In Section~\ref{sec-intentional} we review Naish's approach to
what we call specification semantics.
In Section \ref{sec-generalised} we present a new four-valued approach
which combines two three-valued approaches (Fitting and Naish).
Section~\ref{sec-model-intersect} establishes a
property of this semantics analogous to model intersection.
Section~\ref{sec-proganalysis} shows how four-valued logic 
naturally captures the kind of approximation employed in program
analysis.
Section~\ref{sec-modes} shows how it also helps with modelling the 
concept of \emph{modes} in a moded logic language such as Mercury.
Section~\ref{sec-specification} discusses its relevance for
formal specification
and Section~\ref{sec-debugging} sketches its application to 
declarative debugging.
Section~\ref{sec-computation-information} shows how the logic programming
computation model can be seen in terms of the information ordering.
Section~\ref{sec-related} discusses some additional related work
and Section~\ref{sec-conclusion} concludes.

\section{Logic programs}
\label{sec-scene}

Suppose we need Prolog predicates to capture the workings of
classical propositional disjunction and negation.
We may specify the behaviour exhaustively (we use \texttt{neg} for
negation since \texttt{not} is often used as a general negation primitive in
Prolog):
\pagebreak
\begin{verbatim}
     or(t, t, t).                               neg(t, f).
     or(t, f, t).                               neg(f, t).
     or(f, t, t).
     or(f, f, f).
\end{verbatim}
yielding simple, correct predicates.
If we also need a predicate for implication, we could define
\begin{verbatim}
     implies(X, Y) :- neg(X, U), or(U, Y, t).
\end{verbatim}
Clauses are universally closed.
Stated differently,
the variables in the head of a clause are universally quantified over the
whole clause; those which only occur in the body are existentially
quantified within the body.

Although Prolog programs explicitly define only what is true, it is also
important that they implicitly define what is false.  This is the case
for most programs and is essential when negation is used.  For example,
\texttt{neg(t,t)} would be considered false and for it to succeed would
be an error.  Because (implicit) falsehood depends on the set of all
clauses defining a predicate, it is often convenient to group all
clauses into a single definition with distinct variables in the
arguments of the clause head.  This can be done in Prolog by using the
equality (\texttt{=}) and disjunction (\texttt{;}) primitives.  For
example, \texttt{neg} could be defined 
\begin{verbatim}
     neg(X, Y) :- (X=t, Y=f ; X=f, Y=t).
\end{verbatim}
\citeN{Cla78}
defined the \emph{completion} of a logic program which explicitly
groups clauses together in this way; others 
\cite{fitting:JLP_1991,sem3neg} assume
the program contains a single clause per predicate from the outset.
Henceforth we assume the same.  
The \verb@:-@ in a single-clause definition thus tells us about 
both the truth and falsehood of instances of the head.
Exactly how \verb@:-@ is best viewed has been the topic of much debate
and is a central focus of this paper.  One issue is the relationship
between the truth values of the head and body---what set of truth
values do we use, what constitutes a model or a fixed point, etc.
Another is whether we consider one particular model/fixed point (such
as the least one according to some ordering) as the semantics or do we
consider any one of them to be a possible semantics or consider the set
of all models/fixed points as the semantics.

Let us fix our vocabulary for logic programs and lay down
an abstract syntactic form.
\begin{definition}[Syntax] \rm
An atom (or atomic formula) is of the form
$p(t_1,\ldots,t_n)$, where $p$ is a predicate symbol (of arity $n$)
and $t_1,\ldots,t_n$ are terms.
If $A = p(t_1,\ldots,t_n)$ then $A$'s predicate symbol $\ident{pred}(A)$
is $p$.
There is a distinguished equality predicate $=$ with arity 2, written
using infix notation.
A \emph{literal} is an \emph{atom} $A$ or the negation of an atom,
written $\lnot A$.
A \emph{conjunction} $C$ is a conjunction of literals.
A \emph{disjunction} $D$ is of the form $C_1 \lor \cdots \lor C_k$, 
$k>0$, where each $C_i$ is a conjunction.
For a syntactic object $o$ (literal, clause, disjunction, and so on),
we use $\ident{vars}(o)$ to denote the set of variables that occur
in $o$.

A \emph{predicate definition} is a pair $(H, \exists W[D])$ 
where $H$ is an atom in most general form $p(V_1,\ldots,V_n)$ 
(that is, the $V_i$ are distinct variables), 
$D$ is a disjunction, 
and $W = \ident{vars}(D) \setminus \ident{vars}(H)$.
We call $H$ the \emph{head} of the definition and $\exists W[D]$
its \emph{body}.
The variables $\ident{vars}(H)$ are the \emph{head variables} 
and the variables $W$ are \emph{local variables}.
Finally, a \emph{program} is a finite set $S$ of predicate definitions 
such that if $(H_1,B_1) \in S$ and $(H_2,B_2) \in S$ then 
$\ident{pred}(H_1) \not= \ident{pred}(H_2)$.
\end{definition}

In program text we use Prolog notation and assume this is
converted to the abstract syntax described above by combining
clauses and mapping ``\verb@,@'', ``\verb@;@'' and ``\verb@not@'' to
$\land$, $\lor$ and $\lnot$, respectively, etc.  For example, the
definition of \verb@implies/3@ above is shorthand for
$(implies(X, Y), \exists U [ neg(X, U) \land or(U, Y, t)])$.

We let $\Her$ denote the set of ground (that is, variable-free)
atoms (for some suitably large fixed alphabet).

\begin{definition}[Head instance, head grounding] \rm
\label{def-headinstance}
A \emph{head instance} of a predicate definition $(H,\exists W[D])$
is an instance where all head variables have been replaced by other
terms, and all local variables remain unchanged.
A \emph{head grounding} is a head instance where the head is ground.
\end{definition}

For example, $(implies(t, f), \exists U [ neg(t, U) \land or(U, f, t)])$
is a head grounding of the \verb@implies/3@ definition.
Later we shall define models and  ``immediate consequence'' 
functions for two-, three-, and four-valued semantics.
The use of \emph{head groundings}, rather than more conventional 
approaches is a technical convenience which allows us to emphasize 
the relationship between models and immediate consequence.

\section{Interpretations and models}
\label{sec-interpretations}

In two-valued logic, an interpretation is a mapping from $\Her$
to $\Bi = \{\mathbf{f},\mathbf{t}\}$.
To give meaning to recursively defined predicates, the usual approach is
to impose some structure on $\Her \fun \Bi$, to ensure that we
are dealing with a lattice, or a semi-lattice at least.
Given the traditional ``closed-world'' assumption (that a formula is
false unless it can be proven true), the natural ordering on \Bi\ is this:
$b_1 \leq b_2$ iff $b_1 = \mathbf{f} \lor b_2 = \mathbf{t}$.
The ordering on interpretations is the natural (pointwise)
extension of $\leq$,
equipped with which $\Her \fun \Bi$ is a complete lattice.

Three-valued logic is arguably a more natural logic for the partial
predicates that emerge from pure Prolog programs, and more
generally, for the partial functions that emerge from programming
in any Turing complete language.
The case for three-valued logic as the appropriate logic for
computation has been made repeatedly, starting
with \citeN{Kleene38} and pursued by the VDM school
(for example
\citeN{barringer-cheng-cbjones:1984}, \citeN{jones_middelburg}),
and others.
The third value, $\mathbf{u}$, for ``undefined'', finds natural uses,
for example as the value of \verb!p(b)!, given the program in
Figure~\ref{eg-clark}.
\begin{figure}
\begin{verbatim}
                              p(a).
                              p(b) :- p(b).
                              p(c) :- not p(c).
                              p(d) :- not p(a).
\end{verbatim}
\caption{Small program to exemplify semantics}
\label{eg-clark}
\end{figure}

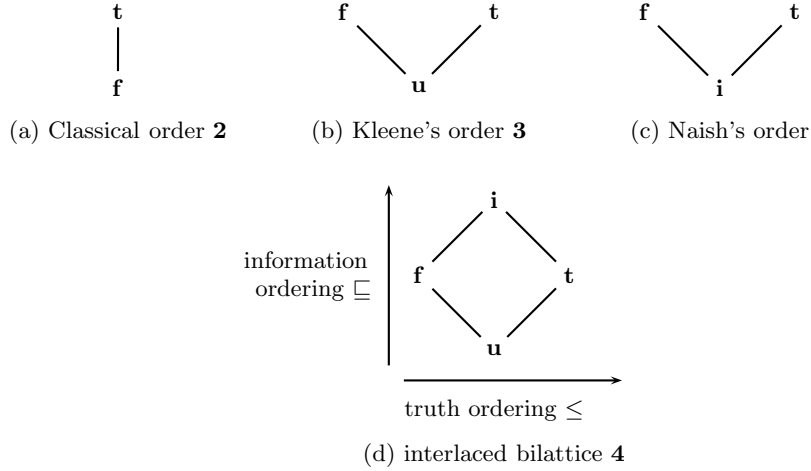
\begin{figure*}[t]
\begin{center}
\begin{pspicture}(-16.3,-5.4)(-3,1)  
\psset{nodesep=3pt}

\rput[c](-14,-0.5){\rnode[c]{F}{\textbf{f}}}
\rput[c](-14,0.5){\rnode[c]{T}{\textbf{t}}}

\ncline{F}{T}

\rput[c](-14,-1.1){(a) Classical order \Bi}

\rput[c](-10,-0.5){\rnode[c]{U3}{\textbf{u}}}
\rput[c](-11,0.5){\rnode[c]{F3}{\textbf{f}}}
\rput[c](-9,0.5){\rnode[c]{T3}{\textbf{t}}}

\ncline{U3}{F3}
\ncline{U3}{T3}

\rput[c](-10,-1.1){(b) Kleene's order \Tri}

\rput[c](-6,-0.5){\rnode[c]{IN}{\textbf{i}}}
\rput[c](-7,0.5){\rnode[c]{FN}{\textbf{f}}}
\rput[c](-5,0.5){\rnode[c]{TN}{\textbf{t}}}

\ncline{IN}{FN}
\ncline{IN}{TN}

\rput[c](-6,-1.1){(c) Naish's order}

\rput[c](-9,-4){\rnode[c]{U4}{\textbf{u}}}
\rput[c](-10,-3){\rnode[c]{F4}{\textbf{f}}}
\rput[c](-8,-3){\rnode[c]{T4}{\textbf{t}}}
\rput[c](-9,-2){\rnode[c]{I4}{\textbf{i}}}

\ncline{U4}{F4}
\ncline{U4}{T4}
\ncline{F4}{I4}
\ncline{T4}{I4}

\rput[c](-11.5,-2.8){information}
\rput[c](-11.4,-3.2){ordering $\sqsubseteq$}
\psline[arrows=->](-10.4,-4.2)(-10.4,-1.8)
\rput[c](-9,-4.8){truth ordering $\leq$}
\psline[arrows=->](-10.2,-4.4)(-7.3,-4.4)

\rput[c](-9,-5.4){(d) interlaced bilattice \Quad}

\end{pspicture}
\end{center}
\caption{Partially ordered sets of truth values\label{fig-bilattice}}
\end{figure*}

With three- or four-valued logic, an interpretation becomes a mapping
from $\Her$ to $\Tri = \{\mathbf{u},\mathbf{f},\mathbf{t}\}$ or to
$\Quad = \{\mathbf{u},\mathbf{f},\mathbf{t},\mathbf{i}\}$
(we discuss the role of the fourth value $\mathbf{i}$ shortly).
For compatibility with the way equality is treated in Prolog, we
constrain interpretations so $x = y$ is mapped to \textbf{t} if $x$ and
$y$ are identical (ground) terms, and to \textbf{f} otherwise. 
This is irrespective of the set of truth values used.
There are different choices for the semantics of the connectives.
In Section~\ref{sec-four} we discuss connectives and give particular
truth tables for the common
connectives, corresponding to Belnap's four-valued 
logic~\cite{Belnap_4val_1977}
(the restriction to three-valued logic that results from deleting rows
and columns containing $\mathbf{i}$ corresponds to
Kleene's (strong) three-valued logic $K_3$ \cite{Kleene38}).

We denote the ordering depicted in Figure~\ref{fig-bilattice}(b) by
$\sqsubseteq$,\footnote{While (b) and (c) are structurally identical, 
\textbf{u} and \textbf{i} carry different meanings, as discussed later.} 
that is, $b_1 \sqsubseteq b_2$ iff
$b_1 = \mathbf{u} \lor b_1 = b_2$, and we overload this symbol
to also denote the ordering in Figure~\ref{fig-bilattice}(d)
(that is, $b_1 \sqsubseteq b_2$ iff
$b_1 = \mathbf{u} \lor b_1 = b_2 \lor b_2 = \mathbf{i}$), as well
of the natural extensions to $\Her \fun \Tri$ or
$\Her \fun \Quad$.
We shall also use $\sqsupseteq$, the inverse of $\sqsubseteq$.
In some contexts we disambiguate the symbol by using a superscript:
$\sqsupseteq^\Tri$ or $\sqsupseteq^\Quad$.  Similarly, we use
$\ge^\Bi$ for the truth ordering with two values, and $=^\Bi$, $=^\Tri$ and $=^\Quad$
for equality of truth values in the different domains.
When the context allows, we write the partially ordered set 
$(\Bi,\leq)$ simply as $\Bi$, $(\Tri,\sqsubseteq^{\Tri})$ as $\Tri$,
and $(\Quad,\sqsubseteq^{\Quad})$ as $\Quad$.

The structure in Figure~\ref{fig-bilattice}(d) is the simplest of
Ginsberg's so-called bilattices \cite{ginsberg:1988}.
The diamond shape can be considered a lattice from two distinct angles.
The ordering $\leq$ is the ``truth'' ordering, whereas $\sqsubseteq$
is the ``information'' ordering.
For the truth ordering we denote the meet and join operations by
$\land$ and $\lor$, respectively.
For the information ordering we denote the meet and join operations by
$\sqcap$ and $\sqcup$, respectively.
Thinking of the four elements as \emph{sets} of classical values,
with $\mathbf{u} = \emptyset$, $\mathbf{i} = \{\mathbf{f},\mathbf{t}\}$,
and $\mathbf{f}$ and $\mathbf{t}$ being singleton sets,
the information ordering is simply the subset ordering.
Regarding the truth ordering, note that
$b_1 \leq b_2$ holds if and only if $b_2$ is at least as true
as $b_1$, and at the same time no more false.
That is, we can move up in the truth value ordering by adding
truth, or removing falsehood, or both.
The bilattice in Figure~\ref{fig-bilattice}(d) is interlaced: 
Each meet and each join operation is monotone with respect to 
\emph{either} ordering.
The bilattice is also distributive in the strong sense that each meet and
each join operation distributes over all the other meet and join 
operations.

An equivalent view of three- or four-valued interpretations is to
consider an interpretation to be a pair of ground atom sets.
That is, the set of interpretations
$\Interp = \pset{\Her} \times \pset{\Her}$.
In this view an interpretation $I = (T_I,F_I)$ is a set $T_I$ of
ground atoms deemed true together with a set $F_I$ of ground atoms 
deemed false.
A ground atom $A$ that appears in neither is deemed undefined.
Such a truth value \emph{gap} may arise from the absence of any
evidence that $A$ should be true, or that $A$ should be false.
In a four-valued setting, para-consistency is a possibility:
A ground atom $A$ may belong to $T_I \cap F_I$.
Such a truth value \emph{glut} may arise from the presence of
conflicting evidence regarding $A$'s truth value.

The concept of a \emph{model} is central to many approaches to logic
programming.   A model is an interpretation which satisfies a particular
relationship between the truth values of the head and body of each head
grounding.  We now define how truth for atoms is lifted to truth for
bodies of definitions.

\begin{definition}[Made true] \rm
Let $I = (T_I, F_I)$ be an interpretation.
Recall that a ground equality atom is in $T_I$ or $F_I$, depending on
whether its arguments are one and the same term.
\\

\noindent
For a ground atom $A$,
\vspace*{-1ex}
\[
\begin{array}{l}
   \mbox{$I$ makes $A$ true iff $A \in T_I$}
\\ \mbox{$I$ makes $A$ false iff $A \in F_I$}
\end{array}
\]
For a ground negated atom $\neg A$,
\vspace*{-1ex}
\[
\begin{array}{l}
   \mbox{$I$ makes $\neg A$ true iff $A \in F_I$}
\\ \mbox{$I$ makes $\neg A$ false iff $A \in T_I$}
\end{array}
\]
For a ground conjunction $C = L_1 \land \cdots \land L_n$,
\vspace*{-1ex}
\[
\begin{array}{l}
   \mbox{$I$ makes $C$ true iff $\forall i \in \{1,\ldots n\}\ I$ makes $L_i$ true}
\\ \mbox{$I$ makes $C$ false iff $\exists i \in \{1,\ldots n\}\ I$ makes $L_i$ false}
\end{array}
\]
For a ground disjunction $D = C_1 \lor \cdots \lor C_n$,
\vspace*{-1ex}
\[
\begin{array}{l}
\mbox{$I$ makes $D$ true iff $\exists i \in \{1,\ldots n\}\ I$ makes $C_i$ true}
\\ \mbox{$I$ makes $D$ false iff $\forall i \in \{1,\ldots n\}\ I$ makes $C_i$ false}
\end{array}
\]
For the existential closure of a disjunction $\exists W [D]$,
\vspace*{-1ex}
\[
\begin{array}{l}
\mbox{$I$ makes $\exists W [D]$ true iff} \\
  \qquad\qquad \mbox{$I$ makes some ground instance of $D$ true} \\
\mbox{$I$ makes $\exists W [D]$ false iff} \\
  \qquad\qquad \mbox{$I$ makes all ground instances of $D$ false} \\
\end{array}
\]
\end{definition}
We use this to extend interpretations naturally so they map $\Her$ and
existential closures of disjunctions to $\Bi$, $\Tri$ or $\Quad$.
We freely switch between viewing an interpretation as a mapping and
as a pair of sets.
Thus, for any formula $F$,
\[
I(F) = 
  \left\{
  \begin{array}{ll}
     \mathbf{u} & \mbox{if $I$ neither makes $F$ true nor false}
  \\ \mathbf{f} & \mbox{if $I$ makes $F$ false and not true}
  \\ \mathbf{t} & \mbox{if $I$ makes $F$ true and not false}
  \\ \mathbf{i} & \mbox{if $I$ makes $F$ true and also false}
  \end{array}
  \right.
\]

\begin{definition}[$\mathcal{R}^\mathcal{D}$-Model] \rm

Let $\mathcal{D}$ be $\Bi$, $\Tri$ or $\Quad$ and $\mathcal{R}^\mathcal{D}$
be a binary relation on $\mathcal{D}$.
An interpretation $I$ is an $\mathcal{R}^\mathcal{D}$-model
of predicate definition $(H,B)$ if, for each head 
grounding $(H\theta,B\theta)$, we have
$\mathcal{R}^\mathcal{D}(I(H\theta),I(B\theta))$.
$I$ is an $\mathcal{R}^\mathcal{D}$-model of program $P$ if it is an
$\mathcal{R}^\mathcal{D}$-model of every predicate definition in $P$.
\end{definition}
For example, a $=^\Bi$-model is a two-valued interpretation where the
head and body of each head grounding have the same truth value.

Another important concept used in logic programming semantics and analysis
is the ``immediate consequence operator''.  The original version, $T_P$,
took a set of true atoms (representing a two-valued interpretation)
and returned the set of atoms which could be proved from those atoms
by using some clause of program $P$ for a single deduction 
step.\footnote{The original version, due to 
\citeN{vanEm-Kow:JACM76}, used `$T$', but $T_P$ has become standard.
Various definitions which generalise $T_P$
to $\Tri$ and $\Quad$ have been given \cite{apt94logic}.}
Here we give an equivalent definition based on how we define 
interpretations.
We write $\Phi_P$ for the immediate consequence operator, following
\citeN{Fitting85}.
Note, however, that we give $\Phi_P$ a definition in terms of
\emph{head groundings} (Definition~\ref{def-headinstance}),
and the same definition is used for the two-, three-, and
four-valued cases alike.

\begin{definition}[$\Phi_P$] \rm
Given an interpretation $I$ and program $P$, $\Phi_P(I)$ is the
interpretation $I'$ such that the truth value of an atom $H$ in $I'$ is
the truth value of $B$ in $I$, where $(H,B)$ is a head grounding of a
definition in $P$.
\end{definition}

\begin{proposition} \rm
Let $\mathcal{D}$ be $\Bi$, $\Tri$ or $\Quad$.
A ($\mathcal{D}$-) interpretation $I$ is a fixed point of $\Phi_P$ 
iff $I$ is a $=^{\mathcal{D}}$-model of $P$.
\end{proposition}
\begin{proof}
This follows easily from the given definitions.
Assume $\Phi_P(I) = I$.
Then, by definition of $\Phi_P(I)$, for each head grounding $(H,B)$
of some predicate definition in $P$, $I(H) =^{\mathcal{D}} I(B)$.
That is, $I$ is a $=^{\mathcal{D}}$-model of $P$.
Conversely, assume $I$ is a $=^{\mathcal{D}}$-model of each predicate
definition $(H,B)$.
That is, $H\theta =^{\mathcal{D}} B\theta$ for all $\theta$.
Then, by definition of $\Phi_P$, $\Phi_P(I) = I$.
\end{proof}

\section{Logic program operational semantics}
\label{sec-lp-semantics}

We first discuss some basic notions and how Clark's two-valued approach
to logic program semantics fits with what we have presented so far.
Then we discuss the Fitting/Kunen three-valued approach and Fitting's
four-valued semantics.

\subsection{Two-valued semantics}
\label{sec-clark}

There are three aspects to the semantics of logic programs: proof
theory, model theory and fixed point theory (see \citeN{Llo84},
for example).  The proof theory is generally based on resolution,
often some variant of SLDNF resolution \cite{Cla78}.  This gives a
top-down operational semantics, which is not our main focus but is
briefly discussed in Section~\ref{sec-computation-information}.
The model theory gives a declarative view of programs and is particularly
useful for high level reasoning about partial correctness.  The fixed
point semantics, based on $\Phi_P$ or $T_P$, gives an alternative 
``bottom up'' operational semantics (which has been used in deductive
databases) and which is also particularly useful for program analysis.

The simplest semantics for pure Prolog disallows negation and treats a
Prolog program as a set of definite clauses.  Prolog's \texttt{:-} is
treated as classical implication, $\leftarrow$, that is, $\ge^\Bi$-models
are used.  There is an important soundness result: if the programmer
has an intended interpretation which is a model, any ground atom which
succeeds is true in that model.  The ($\le$) least model is also the
least $=^\Bi$-model and the least fixed point of $\Phi_P$, which is
monotone in the truth ordering (so a least fixed point always exists).
The set of true atoms in this least model is the set of atoms which are
true in all $\ge^\Bi$-models (and $=^\Bi$-models) and is also the set of
atoms which have a successful derivation using SLD resolution. 
For these reasons, this is the accepted
semantics for Prolog programs without negation.

To support negation in the semantics, \citeN{Cla78} combined all
clauses defining a particular predicate into a single ``if and only if''
definition which uses the classical bi-implication $\leftrightarrow$.
This is called the Clark completion $comp(P)$ of a program $P$.  
Our definitions
are essentially the same, but we avoid the $\leftrightarrow$ symbol.
In this paper's terminology,
Clark used $=^\Bi$-models, which correspond to classical fixed points 
of $\Phi_P$.  Clark specifically considered logical consequences
of $comp(P)$: atoms which were true in all $=^\Bi$-models.

The soundness result above applies, and any finitely failed ground atom
must also be false in the programmer's intended interpretation, if 
this interpretation is a model.  
However, because $\Phi_P$ is non-monotone in the
truth ordering when negation is present, there may be multiple minimal
fixed points/models, or there may be none.  For example, using Clark's
semantics for the program in Figure \ref{eg-clark}, there is no model
and no fixed point due to the clause for \texttt{p(c)}, yet the query
\texttt{p(a)} succeeds and \texttt{p(d)} finitely fails.  Thus the Clark
semantics does not align particularly well with the operational semantics.

\subsection{Three-valued semantics}
\label{sec-fitting-kunen}

Even in the absence of negation, a two-valued semantics is lacking in
its inability to distinguish failure and looping.
\citeN{Mycroft} explored the use of many-valued logics,
including \Tri, to remedy this.
Mycroft discussed this for Horn clause programs, and
others, including \citeN{Fitting85} and \citeN{Kunen87},
subsequently adapted Clark's work to a three-valued logic, 
addressing the problem of how to account properly for the use of
explicit negation in programs.

In a two-valued setting, the Clark completion may be inconsistent,
witness the completion of the clause for \verb!p(c)! in
Figure~\ref{eg-clark}.
Hence the Clark completion is unable to give a reasonable meaning to
\texttt{p(a)}, \texttt{p(b)}, and \texttt{p(d)}, even though these atoms
do not depend on \texttt{p(c)}.
If we were to delete the clause for \texttt{p(c)} in Figure
\ref{eg-clark}, the Clark semantics would map \texttt{p(b)} to \textbf{f},
even though it does not finitely fail.   
The reason is that the smallest 2-valued model of the Clark completion
$\mathtt{p(b)} \Leftrightarrow \mathtt{p(b)}$ maps \texttt{p(b)} to
\textbf{f}.

However, a $=^\Tri$-model always exists for a Clark-completed program;
for example, \texttt{p(c)} takes on the third truth value.  
Moreover, since $\Phi_P$ is monotone with respect to the information
ordering, a least fixed point always exists and coincides with
the least $=^{\Tri}$-model.
Ground atoms which are \textbf{t} in this model 
(such as \texttt{p(a)} in Figure~\ref{eg-clark})
are those which have successful derivations, while 
ground atoms which are \textbf{f} (such as \texttt{p(d)}) 
are those which have finitely failed SLDNF trees~\cite{Cla78}.
Atoms with the third truth value (\texttt{p(b)} and \texttt{p(c)})
must loop.  
Atoms which are \textbf{t}
or \textbf{f} in the Fitting/Kunen semantics may also loop if the search
strategy or computation rule are unfair (even without negation, \textbf{t}
atoms may loop with an unfair search strategy).  Furthermore, when
negation is present, a computation may \emph{flounder} owing to a negated
call which never becomes ground and hence is never selected (this is a fourth
possible behaviour).  However the Fitting/Kunen
approach does align the model theoretic and fixed point semantics much more
closely to the operational semantics of Prolog than the approach of
Clark, and we can imagine an idealised logic programming language where
the alignment is precise.

$\Phi_P$ has a drawback, though: while monotone, it is not in general 
continuous.
\citeN{Blair} shows that the smallest ordinal $\beta$ for which 
$\Phi_P^\beta(\bot)$ is the least fixed point of $\Phi_P$ may not be
recursive\footnote{The (possibly transfinite) powers of $\Phi_P$ are
defined the standard way: For a successor ordinal $\beta$,
$\Phi_P^{\beta}(x) = \Phi_P(\Phi_P^{\beta - 1}(x))$, 
and for a limit ordinal $\beta$, 
$\Phi_P^{\beta}(x) = \bigsqcup _{\alpha < \beta} \Phi_P^{\alpha}(x)$.}
\citeN{Kunen87} shows that, with a semantics based on three-valued
Herbrand models (all models or the least model), the set of ground
atoms true in such models may not be recursively enumerable\footnote{%
We use $\bot$ to denote the smallest interpretation with respect to 
$\sqsubseteq$.}.
Kunen instead suggests a semantics based on \emph{any} three-valued
model and shows that truth (\textbf{t}) in all $=^\Tri$-models is
equivalent to being deemed true by $\Phi_P^n(\bot)$ for some $n \in \nat$.
Hence Kunen proposes $\Phi_P^\omega(\bot)$ as the meaning of program $P$.
For a given $P$ and ground atom $A$, it is decidable whether
$A$ is \textbf{t} in $\Phi_P^n(\bot)$, so whether $A$ is \textbf{t} in 
$\Phi_P^\omega(\bot)$ is semi-decidable.

For simplicity, in this paper we take (the possibly non-computable) 
$M = \ident{lfp}(\Phi_P)$ to be the meaning of a program, that is, the
least $=^\Tri$-model.
However, since we shall be concerned with over-approximations to
$M$, what we shall have to say will apply equally well if Kunen's
$\Phi_P^\omega(\bot)$ is assumed.

\subsection{Four-valued semantics}
\label{sec-four}

Subsequent to his three-valued proposal, Fitting recommended,
in a series of papers 
\citeyear{fitting:Fund_Info_1988,fitting:LICS1989,fitting:JLC_1991,fitting:JLP_1991,fitting:TCS_2002},
bilattices as suitable bases for logic program semantics.
The bilattice \Quad\ (Figure~\ref{fig-bilattice}(d)) was just one of 
several studied for the purpose, and arguably the most important one.

Fitting's motivation for employing four-valued logic was, apart from
the elegance of the interlaced bilattices and their algebraic
properties, the application in a logic programming language which
supports a notion of (spatially) distributed programs. 
In this setting there is a natural need for a fourth truth value, $\top$
(our \textbf{i}), to denote conflicting information received from
different nodes in a computing network.

In the language proposed by \citeN{fitting:JLP_1991}, 
the traditional logical connectives used on the 
right-hand sides of predicate definitions are explained in terms 
of the truth ordering.
Negation is reflection in the truth ordering: 
$\neg \textbf{u} = \textbf{u}$, $\neg \textbf{f} = \textbf{t}$, 
$\neg \textbf{t} = \textbf{f}$ and $\neg \textbf{i} = \textbf{i}$,
conjunction is meet ($\land$), disjunction is join ($\lor$), and 
existential quantification is the least upper bound ($\bigvee$) of 
all instances.
These tables give conjunction, disjunction and negation in \Quad:
\medskip

\begin{center}
\begin{minipage}{0.38\textwidth}
\begin{tabular}{|c||c|c|c|c|}
\cline{1-5}
$\wedge$ & \textbf{u} & \textbf{t} & \textbf{f} & \textbf{i}\\
\cline{1-5}
\vspace{-3.9mm} & & & & \\
\cline{1-5}
\textbf{u}          & \textbf{u} & \textbf{u} & \textbf{f} & \textbf{f}\\
\cline{1-5}
\textbf{t}          & \textbf{u} & \textbf{t} & \textbf{f} & \textbf{i}\\
\cline{1-5}
\textbf{f}          & \textbf{f} & \textbf{f} & \textbf{f} & \textbf{f}\\
\cline{1-5}
\textbf{i}          & \textbf{f} & \textbf{i} & \textbf{f} & \textbf{i}\\
\cline{1-5}
\end{tabular}
\end{minipage}
\begin{minipage}{0.38\textwidth}
\begin{tabular}{|c||c|c|c|c|}
\cline{1-5}
$\vee$ & \textbf{u} & \textbf{t} & \textbf{f} & \textbf{i}\\
\cline{1-5}
\vspace{-3.9mm} & & & & \\
\cline{1-5}
\textbf{u}        & \textbf{u} & \textbf{t} & \textbf{u} & \textbf{t}\\
\cline{1-5}
\textbf{t}        & \textbf{t} & \textbf{t} & \textbf{t} & \textbf{t}\\
\cline{1-5}
\textbf{f}        & \textbf{u} & \textbf{t} & \textbf{f} & \textbf{i}\\
\cline{1-5}
\textbf{i}        & \textbf{t} & \textbf{t} & \textbf{i} & \textbf{i}\\
\cline{1-5}
\end{tabular}
\end{minipage}
\begin{minipage}{0.19\textwidth}
\begin{tabular}{|c||c|}
\cline{1-2}
   $\neg$ &
\\ \cline{1-2}
\vspace{-3.9mm} &
\\ \cline{1-2}
   \textbf{u} & \textbf{u}
\\ \cline{1-2}
   \textbf{f} & \textbf{t}
\\ \cline{1-2}
   \textbf{t} & \textbf{f}
\\ \cline{1-2}
   \textbf{i} & \textbf{i}
\\ \cline{1-2}
\end{tabular}
\end{minipage}
\medskip
\end{center}

\noindent
The operations $\sqcap$ and $\sqcup$ are similarly given by 
Figure~\ref{fig-bilattice}(d).
Fitting refers to $\sqcap$ (he writes $\otimes$) as \emph{consensus},
since $x \sqcap y$ represents what $x$ and $y$ agree about.
The $\sqcup$ operation (which he writes as $\oplus$) he refers to as
\emph{gullibility}, since $x \sqcup y$ represents agreement with both
$x$ and $y$, whatever they say, including cases where they disagree.
\citeN{Palmer97} also uses this logic with another parallel logic
programming language, Andorra Kernel Language.  Although AKL does not
support $\sqcap$ or $\sqcup$ as explicit language primitives, Palmer's
AKL compiler uses such operations in its analysis of parallel
sub-computations which may or may not agree on their results.

The idea of an information (or knowledge) ordering is familiar to anybody
who has used domain theory and denotational semantics.
To give meaning to recursively defined objects we refer to fixed points
of functions defined on structures equipped with some ordering---the
information ordering.
This happens already in the three-valued approaches to semantics
discussed above. 
Three-valued semantics does use the
distinction between a truth ordering $\leq$ and an information ordering
$\sqsubseteq$, but it does not expose it as radically as the bilattice.
In Fitting's words, the three-valued approach, ``while abstracting away
some of the details of [Kripke's theory of truth] still hides the
double ordering structure''~\cite{fitting:2006}.

The logic programming language of \citeN{fitting:JLP_1991} contains
operators $\otimes$ and $\oplus$, reflecting the motivation in terms of
distributed programs.
We, on the other hand, deal with a language with traditional pure 
Prolog syntax. 
If the task was simply to model its operational semantics,
having four truth values rather than three would offer little, if any,
advantage.
However, our motivation for using four-valued logic is very different
to that of Fitting.
We find compelling reasons for the use of four-valued logic 
to explain certain programming language features, as well as to
embrace, semantically, such software engineering aspects as
program correctness with respect to programmer intent or specification,
declarative debugging, and program analysis.
We next discuss one of these aspects.

\section{Three-valued specification semantics}
\label{sec-intentional}

\citeN{sem3neg} proposed an alternative three-valued semantics.
Unlike other approaches, the objective was not to align declarative and
operational semantics.  Instead, the aim was to provide a declarative
semantics which can help programmers develop correct code in a natural
way.  Naish argued that intentions of programmers are not two-valued.  It
is generally intended that some ground atoms should succeed (be considered
\textbf{t}) and some should finitely fail (be considered \textbf{f}) but
some should never occur in practice; there is no particular intention
for how they should behave and the programmer does not care and often
does not know how they behave. An example is merging lists, where it is
assumed two sorted lists are given as input: it may be more appropriate to
consider the value of \verb!merge([3,2],[1],[1,3,2])! \emph{irrelevant}
than to give it a classical truth value, since a \emph{precondition}
is violated.  Or consider this program:
\begin{verbatim}
     or2(t, _, t).                             or3(_, t, t).
     or2(f, B, B).                             or3(B, f, B).
\end{verbatim}
It gives two alternative definitions of \texttt{or} 
(previously defined in Section~\ref{sec-scene}), both designed with the 
assumption that the first two arguments will always be Booleans.
If they are not, we consider the atom to be \emph{inadmissible}
(a term used in debugging \cite{Per86,ddscheme3}) and give it
the truth value \textbf{i}.
Interpretations can be thought of as the programmer's understanding
of a specification, where \textbf{i} is used for underspecification
of behaviour.  The same three-valued interpretation can then be used with
all three definitions of \texttt{or}.  A programmer can first fix the
interpretation then code any of these definitions and reason about their
correctness.  In contrast, both the Clark and Fitting/Kunen semantics
assign different, incompatible
meanings to the three definitions, with atoms such
as \verb@or3(4,f,4)@ and \verb@or2(t,[],t)@ considered \textbf{t} and
\verb@or3(t,[],t)@ considered \textbf{f}.  In order for the programmer's
intended interpretation to be a $=^\Bi$-model or $=^\Tri$-model, unnatural
distinctions such as these must be made.  \citeN{sem3neg} argues
that it is unrealistic for programmers to use such interpretations as 
a basis for reasoning about correctness of their programs.
In Section~\ref{sec-generalised} we consider a somewhat larger example
in more depth.

Although Naish uses \textbf{i} instead of \textbf{u} as the third truth
value, his approach is structurally the same as Fitting/Kunen's with
respect to the ordering, Figure~\ref{fig-bilattice}(b) and (c), the
$\Phi_P$ operator and the meaning of connectives used in
the body of definitions.  The key technical
difference is how Prolog's \texttt{:-}
is interpreted.  Fitting generalises Clark's classical $\leftrightarrow$
to $\cong$ or ``strong equivalence'', where heads and bodies of head
groundings must have the same truth values.
Naish defined a different ``arrow'', $\leftarrow$,
which is asymmetric, but not a conservative extension of classical
implication (so the choice of symbol is perhaps misleading).
In addition to identical truth values for heads and
bodies, Naish allows head groundings of the form $(\mathbf{i}, \mathbf{f})$
and $(\mathbf{i}, \mathbf{t})$.
The difference is captured by these tables (Fitting left, Naish
right):\footnote{We abuse notation here: $\cong$ and $\leftarrow$ are not
actually used as connectives, so the table entries should really be ``model''
and ``not model'' rather than \textbf{t} and \textbf{f}.}
\smallskip

\begin{center}
\begin{minipage}{0.48\textwidth}
\begin{tabular}{|c||c|c|c|}
\cline{1-4}
$\cong$ & \textbf{t}  & \textbf{f} & \textbf{u} \\
\cline{1-4}
\vspace{-3.9mm} & & & \\
\cline{1-4}
\textbf{t}   & \textbf{t}  & \textbf{f} & \textbf{f} \\
\cline{1-4}
\textbf{f}   & \textbf{f}  & \textbf{t} & \textbf{f} \\
\cline{1-4}
\textbf{u}   & \textbf{f}  & \textbf{f} & \textbf{t} \\
\cline{1-4}
\end{tabular}
\end{minipage}
\begin{minipage}{0.48\textwidth}
\begin{tabular}{|c||c|c|c|}
\cline{1-4}
$\leftarrow$ & \textbf{t}  & \textbf{f} & \textbf{i} \\
\cline{1-4}
\vspace{-3.9mm} & & & \\
\cline{1-4}
\textbf{t}   & \textbf{t}  & \textbf{f} & \textbf{f} \\
\cline{1-4}
\textbf{f}   & \textbf{f}  & \textbf{t} & \textbf{f} \\
\cline{1-4}
\textbf{i}   & \textbf{t}  & \textbf{t} & \textbf{t} \\
\cline{1-4}
\end{tabular}
\end{minipage}
\end{center}

\vspace*{0.6em}\noindent
Naish's arrow captures the principle that, 
if a predicate is called in an inadmissible way, 
then it does not matter if it succeeds or fails.  
The definition of
a model uses this weaker ``arrow''; we discuss it further in Section
\ref{sec-generalised}.  \citeN{sem3neg} shows that for any model,
only \textbf{t} and \textbf{i} atoms can succeed and only \textbf{f}
and \textbf{i} atoms can finitely fail.  In models of the code in
Figure \ref{eg-clark}, \texttt{p(b)} can be \textbf{t} or \textbf{f}
or \textbf{i} but \texttt{p(c)} can only be \textbf{i}.  For practical
code, programmers can reason about partial correctness using intuitive
models in which the behaviour of some atoms is unspecified.

\section{Four-valued specification semantics}
\label{sec-generalised}

The Fitting/Kunen and Naish approaches all use three truth values,
the Kleene strong three-valued logic for the connectives in the body of
definitions, and the same immediate consequence operator.  It is thus
tempting to assume that the ``third'' truth value in these approaches
is the same in some sense.  This is implicitly assumed by 
\citeN{sem3neg}
when he compares different approaches.
However, the third truth value is used for very different purposes in
the approaches being compared.  
\citeN{Fitting85} and \citeN{Kunen87} use it to make the semantics 
more precise than Clark---distinguishing success and finite failure 
from nontermination (neither success nor finite failure).  
\citeN{sem3neg} uses it to make the semantics
\emph{less} precise than Clark, allowing a truth value corresponding
to success or finite failure.  Thus we believe it is best to treat the
third truth values of Fitting/Kunen and Naish as \emph{duals} instead 
of the same value. Naish treats \textbf{i} as the bottom element
whereas in $\Quad$ it is more naturally the top element, with the
ordering of Naish inverted.
Because conjunction, disjunction and negation in \Quad\ are
symmetric in the information order, the third value in the Kleene strong
three-valued logic can map to either the top or bottom element in \Quad.
This is why the third truth values in Fitting/Kunen and Naish are treated
identically, even though they are better viewed as semantically distinct.

The four values \textbf{t}, \textbf{f}, \textbf{i} and \textbf{u} are
associated with truth/success, falsehood/finite failure, inadmissibility
(the Naish third value) and looping/error (the Fitting/Kunen third value).
Inadmissibility can be seen as saying that both success and failure are
correct, so we can see it as the union of both.  Atoms which are
\textbf{u} in the Fitting/Kunen semantics neither succeed nor finitely fail.  
Thus, as already pointed out,
the information ordering can also be seen as the
set ordering, $\subseteq$, if we interpret the truth values in $\Quad$
as sets of Boolean values.  

\begin{figure}
\begin{verbatim}
            % Checks A-B = E-F, where all are natural numbers,
            % represented in Peano style with 0 and s/1
            % This definition is common to programs P1-P4
            eq_diff(A, B, E, F) :- sub(A, B, D), sub(E, F, D).

% sub/3 definition for P1               % sub/3 definition for P3
sub(0, 0, 0).                           sub(A, A, 0).
sub(s(A), 0, s(D)) :- sub(A, 0, D).     sub(A, B, s(D)) :-
sub(s(A), s(B), D) :- sub(A, B, D).             not(A=B), sub(A, s(B), D).

% sub/3 definition for P2               % sub/3 definition for P4
sub(A, 0, A).                           sub(A, A, 0).
sub(s(A), s(B), D) :- sub(A, B, D).     sub(A, B, s(D)) :- sub(A, s(B), D).
\end{verbatim}
\caption{Programs P1--P4 for subtraction over natural numbers}
\label{fig-sub}
\end{figure}

Consider the four programs depicted in Figure \ref{fig-sub}, which have
a common definition of \verb@eq_diff/4@ (which checks if the differences
of two pairs of natural numbers are the same) but different definitions
of \verb@sub/3@ (which performs subtraction over
natural numbers).  These programs have different sets of ground atoms which
succeed and finitely fail; we discuss these in more detail
later.  For some of the programs there are ground atoms
which neither succeed nor finitely fail and the Fitting/Kunen semantics
has the advantage of reflecting this whereas Clark's
cannot.  For example, the same set of atoms succeed in P3 and P4
(the least two-valued models, used by Clark, exist and are the same) but
some atoms such as \verb@eq_diff(s(0),0,0,s(0))@ finitely fail in P3
but loop in P4 (so the least three-valued models differ).  Naish's
semantics has the advantage of allowing the programmer to reason about
correctness with respect to intentions or specifications which are
imprecise.  For example, a programmer may only want to specify the
desired behaviour of ground atoms \verb@eq_diff(A,B,E,F)@ where all
arguments are natural numbers (of the form $s^n(0)$) and \verb@A-B@ and
\verb@E-F@ are defined over natural numbers; other atoms can reasonably
be considered inadmissible. This allows us to simply establish partial
correctness of all four programs --- the behaviours differ, but only
for atoms the programmer does not care about.  The four-valued semantics
which we propose here combines the advantages of the Fitting/Kunen and
Naish approaches within a single unified framework.

We now show how Naish's semantics can be combined with that of
Fitting/Kunen and generalised to \Quad.
Adding the truth value \textbf{i} is a conservative extension to the
Fitting/Kunen semantics.  The three-valued fixed points of $\Phi_P$ are
preserved, including the least fixed point, and the
way the semantics describes what is \emph{computed} is unchanged,
though the additional truth value can be useful for \emph{approximating}
what is computed.
However, adding the truth value \textbf{u} to the Naish semantics
\emph{does} allow us to describe more precisely  what is \emph{intended}.
There are occasions when both the success and finite failure of an atom
are considered incorrect behaviour and thus \textbf{u} is an appropriate
value to use in the intended interpretation.  We give three examples.
The first is an interpreter for a Turing-complete language.  If the
interpreter is given (the representation of) a looping program it should
not succeed and it should not fail.  The second is an operating system.
Ignoring the details of how interaction with the real world is modelled
in the language, termination means the operating system crashes.
The third is code which is only intended to be called in limited ways,
but is expected to be robust and check its inputs are well formed.
Exceptions or abnormal termination with an error message are best not
considered success or finite failure.  Treating them in the same way as
infinite loops in the semantics may not be ideal but it is more expressive
than using the other three truth values (indeed, ``infinite'' loops are
never really infinite because resources are finite and hence some form
of abnormal termination results).

\citeN{sem3neg}
defines models in terms of the $\leftarrow$ described earlier and shows
that $I$ is a model if and only if 
$I \sqsupseteq^\Tri \Phi_P(I)$.
The significance of this proposition is not noted by \citeN{sem3neg},
but it prompts a key observation: the $\leftarrow$
defines the information order on truth values!  The classical arrow
defines the truth ordering on two values; Naish's arrow defines the
orthogonal ordering in the three-valued extension.  It is therefore
clear how Naish's arrow can be generalised to $\Quad$.
The models of
\citeN{sem3neg} are $\sqsupseteq^\Tri$-models, which can be
generalised to $\sqsupseteq^\Quad$-models, and Naish's arrow
is generalised as $\sqsupseteq^\Quad$ (treating both as connectives),
with the following truth table:

\begin{center}
\begin{tabular}{|c||c|c|c|c|}
\cline{1-5}
$\sqsupseteq^\Quad$ & \textbf{u} & \textbf{t} & \textbf{f} & \textbf{i}\\
\cline{1-5}
\vspace{-3.9mm} & & & & \\
\cline{1-5}
\textbf{u}          & \textbf{t} & \textbf{f} & \textbf{f} & \textbf{f}\\
\cline{1-5}
\textbf{t}          & \textbf{t} & \textbf{t} & \textbf{f} & \textbf{f}\\
\cline{1-5}
\textbf{f}          & \textbf{t} & \textbf{f} & \textbf{t} & \textbf{f}\\
\cline{1-5}
\textbf{i}          & \textbf{t} & \textbf{t} & \textbf{t} & \textbf{t}\\
\cline{1-5}
\end{tabular}
\end{center}

\begin{proposition} \rm
\label{prop-fourval-sqsubset}
$I$ is a $\sqsupseteq^\Quad$-model of $P$ iff $\Phi_P(I) \sqsubseteq I$.
\end{proposition}

\begin{proof}
$I$ is a $\sqsupseteq^\Quad$-model iff,
for every head grounding $(H, B)$ of $P$, $I(B) \sqsubseteq I(H)$.  
This is equivalent to stating that 
if $I$ makes $B$ true then $I$ makes $H$ true, and also,
if $I$ makes $B$ false then $I$ makes $H$ false.
But this is the case iff $\Phi_P(I) \sqsubseteq I$,
by the definition of $\Phi_P$.
\end{proof}
It is easy to see that if $I$ is a $\sqsupseteq^\Tri$-model of $P$ 
then $I$ is a $\sqsupseteq^\Quad$-model of $P$.
However, the converse is not 
necessarily true, so the results of \citeN{sem3neg}
cannot be used to show properties of four-valued models.  
However, such properties can be proved directly,
using properties of the lattice of interpretations.

\begin{proposition} \rm
\label{prop-lfp-lmod}
The least $\sqsupseteq^\Quad$-model of $P$ is $lfp(\Phi_P)$.
\end{proposition}
\begin{proof}
By the definition of $\sqsupseteq^\Quad$-model, $I$ is a 
$\sqsupseteq^\Quad$-model iff $\Phi_P(I) \sqsubseteq I$.
Since $\Phi_P$ is monotone,
the Knaster-Tarski theorem \cite{Tarski} establishes
$lfp(\Phi_P)$ as the least $I$ such that $\Phi_P(I) \sqsubseteq I$.
Hence $lfp(\Phi_P)$ is the least $\sqsupseteq^\Quad$-model of $P$.
\end{proof}
For reasoning about partial correctness, the relationship between truth
values in an interpretation and operational behaviour is crucial.

\begin{theorem} \rm
\label{thm-sound-ge}
If $I \sqsupseteq^\Quad lfp(\Phi_P)$ then no \textbf{t} atoms in $I$ can
finitely fail, no \textbf{f} atoms in $I$ can succeed, and no \textbf{u}
atoms in $I$ can finitely fail or succeed.  
\end{theorem}
\begin{proof}
The least fixed point in the four-valued case is the same as
the least fixed point in the three-valued case.
Hence \cite{Kunen87}
finitely failed atoms are \textbf{f} in $lfp(\Phi_P)$, 
successful atoms are \textbf{t} in $lfp(\Phi_P)$,
and \textbf{u} atoms in $lfp(\Phi_P)$ must loop.
From the $\sqsubseteq$ ordering, an atom mapped to \textbf{f} by $I$ 
can only be mapped to \textbf{f} or \textbf{u} by $lfp(\Phi_P)$.
Similarly, atoms which $I$ maps to \textbf{t} can only be 
mapped to \textbf{t} or \textbf{u} by $lfp(\Phi_P)$, and \textbf{u} 
atoms can only be mapped to \textbf{u}.
\end{proof}

\begin{corollary} \rm
\label{thm-sound}
If $I$ is a $\sqsupseteq^\Quad$-model of $P$ then no \textbf{t} atoms in $I$ can
finitely fail, no \textbf{f} atoms in $I$ can succeed and no \textbf{u}
atoms in $I$ can finitely fail or succeed.  
\end{corollary}
\begin{proof}
From Theorem \ref{thm-sound-ge} and Proposition \ref{prop-lfp-lmod}.
\end{proof}
These results about the behaviour of \textbf{t} and \textbf{f} atoms are
essentially the two soundness theorems, for finite failure and success,
respectively, of \citeN{sem3neg}.  The result for \textbf{u} atoms is new.
The relationship between the (idealised) operational semantics and various
forms of four-valued model-theoretic semantics can be summarised by
the following Table (the last row summarises Corollary~\ref{thm-sound}).
It is a refinement of Table 1 of \citeN{sem3neg}, which is the same except
it uses three-valued models and conflates \textbf{i} and \textbf{u}.
\medskip

\begin{center}
\begin{tabular}{|l|c|c|c|}
\cline{1-4}
\textbf{operational} & \textbf{succeed} & \textbf{loop} & \textbf{fail} \\
\cline{1-4}
least $=^\Quad$-model & \textbf{t} & \textbf{u} & \textbf{f} \\
\cline{1-4}
any $=^\Quad$-model & \textbf{t} &
		\textbf{t}/\textbf{u}/\textbf{i}/\textbf{f} & \textbf{f} \\
\cline{1-4}
any $\sqsupseteq^\Quad$-model & \textbf{t}/\textbf{i} &
		\textbf{t}/\textbf{u}/\textbf{i}/\textbf{f} & \textbf{i}/\textbf{f} \\
\cline{1-4}
\cline{1-4}
\end{tabular}
\end{center}

\begin{table}
\begin{tabular}{|l|c|c|c|c|c|c|c|}
\cline{1-8}
                                   &$I_0$&$I_1$&$I_2$&$I_3$&$I_3'$&$I_3''$&$I_4$\\
\cline{1-8}
\verb@eq_diff(s(0),0,s(s(0)),s(0))@ & \textbf{t}  & \textbf{t}  & \textbf{t}  & \textbf{t}  & \textbf{t}  & \textbf{t}  & \textbf{t}\\
\verb@eq_diff(s(0),0,0,0)@          & \textbf{f}  & \textbf{f}  & \textbf{f}  & \textbf{f}  & \textbf{f}  & \textbf{f}  & \textbf{u}\\
\verb@eq_diff([],[],[],[])@         & \textbf{i}  & \textbf{f}  & \textbf{f}  & \textbf{t}  & \textbf{t}  & \textbf{t}  & \textbf{t}\\
\verb@eq_diff([],0,[],0)@           & \textbf{i}  & \textbf{f}  & \textbf{t}  & \textbf{u}  & \textbf{f}  & \textbf{t}  & \textbf{u}\\
\verb@eq_diff(s(0),0,0,s(0))@       & \textbf{i}  & \textbf{f}  & \textbf{f}  & \textbf{f}  & \textbf{f}  & \textbf{f}  & \textbf{u}\\
\verb@eq_diff(0,s(0),0,s(0))@       & \textbf{i}  & \textbf{f}  & \textbf{f}  & \textbf{u}  & \textbf{f}  & \textbf{t}  & \textbf{u}\\
\cline{1-8}
P1 least model                      &    & \chk  &    &    &    &    &  \\
P2 least model                      &    &    & \chk  &    &    &    &  \\
P3 least model                      &    &    &    & \chk  &    &    &  \\
P4 least model                      &    &    &    &    &    &    & \chk\\
\cline{1-8}
P1 $=^\Quad$-model                  &    & \chk  &    &    &    &    &  \\
P2 $=^\Quad$-model                  &    &    & \chk  &    &    &    &  \\
P3 $=^\Quad$-model                  &    &    &    & \chk  & \chk  & \chk  &  \\
P4 $=^\Quad$-model                  &    &    &    &    &    &    & \chk\\
\cline{1-8}
P1 $\sqsupseteq^\Quad$-model        & \chk  & \chk  &    &    &    &    &  \\
P2 $\sqsupseteq^\Quad$-model        & \chk  &    & \chk  &    &    &    &  \\
P3 $\sqsupseteq^\Quad$-model        & \chk  &    &    & \chk  & \chk  & \chk  &  \\
P4 $\sqsupseteq^\Quad$-model        & \chk  &    &    & \chk  & \chk  & \chk  & \chk\\
\cline{1-8}
\end{tabular}
\caption{Seven interpretations of programs P1--P4 from Figure~\ref{fig-sub}}
\label{fig-sub-models}
\end{table}

\vspace{3ex}
\noindent
Consider again the four programs depicted in Figure \ref{fig-sub}.
Table~\ref{fig-sub-models} describes seven interpretations for these
programs.  $I_0$ is the typically intended interpretation, 
with inadmissibility of \verb@eq_diff/4@ defined as before, 
\verb@sub(A,B,C)@ inadmissible
if \verb@A@ or \verb@B@ are not natural numbers or \verb@B>A@, and other
atoms partitioned into \textbf{t} and \textbf{f} in the intuitive way.
$I_1$, $I_2$, $I_3$ and $I_4$ are the least four-valued models of the programs
$P1$--$P4$, respectively.  The truth values in these interpretations also
align with the operational behaviour of the atoms in Prolog. $I_3'$ and
$I_3''$ are the same as $I_3$ except that atoms which are $\mathbf{u}$ in
$I_3$ are $\mathbf{f}$ and $\mathbf{t}$ in $I_3'$ and $I_3''$, respectively.
The top section of Table~\ref{fig-sub-models} gives the truth values
of several representative atoms for each interpretation; we assume the
existence of constant \verb@[]@ in the set of function symbols to show
behaviour of ``ill-typed'' atoms.  $I_1$, $I_3'$ and $I_3''$ are two-valued,
with $I_1$ the least two-valued model (using the truth ordering) of $P1$
and $I_3'$ the least two-valued model of both $P3$ and $P4$.

The later parts of Table~\ref{fig-sub-models} give which of these
interpretations are certain kinds of four-valued models for the different
programs.  The four least models of the respective programs are distinct,
reflecting the different behaviours.  $I_3'$ and $I_3''$ are not the least
model of $P3$ but they are $=^\Quad$-models.  

Note carefully that the intended interpretation,
$I_0$, is a $\sqsupseteq^\Quad$-model of all programs.  
For the four
interpretations shown which are $\sqsupseteq^\Quad$-models of $P3$,
we have $I_0 \sqsupseteq I_3' \sqsupseteq I_3$ and $I_0 \sqsupseteq I_3''
\sqsupseteq I_3$ with $I_3'$ and $I_3''$ incomparable in the information
order.  $P4$ also has all these interpretations as models, along with
$I_4$, which is below $I_3$.  

Also note how we use the two non-classical values
in $\Quad$, that is, \textbf{u} and \textbf{i}, 
for quite distinct purposes.
The two- and three-valued approaches to semantics do not allow such a
complete picture of how these programs behave, along with the ways they
can be viewed by programmers.

The use of $\sqsupseteq^\Quad$-models
allows simple and intuitive verification of partial correctness of
all programs but does not distinguish between total correctness ($P1$--$P3$)
and only partial correctness ($P4$).  However, even analysis of least
models does not guarantee total correctness for Prolog programs because
alignment of truth values and behaviour assumes fairness of the search
strategy (for success) and the computation rule (for finite failure)
and non-floundering whenever negation is present.  For example, if we
reverse the order of the two sub-goals in the \verb@sub/3@ definition
in $P3$ then with Prolog's normal left to right computation rule, $P3$
behaves the same as $P4$ for the atoms shown.

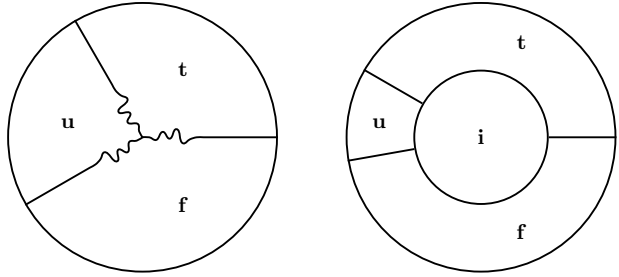
\begin{figure}[t]
\scalebox{0.9}{
\begin{pspicture}(-2,-2)(7,2.2)   
\pscircle(0,0){2}
\pscustom{%
\pscurve%
(0.0,0.0)
(0.1,0.0)
(0.2,-0.04)
(0.3,0.08)
(0.4,-0.02)
(0.5,0.12)
(0.6,-0.09)
(0.7,-0.03)
(0.8,0.0)
(0.9,0.0)
(1.0,0.0)
(1.1,0.0)
(1.2,0.0)
(1.3,0.0)
(1.4,0.0)
(1.5,0.0)
(1.6,0.0)
(1.7,0.0)
(1.8,0.0)
(1.9,0.0)
(2.0,0.0)
\rotate{120}
\pscurve[liftpen=2]%
(0.0,0.0)
(0.1,0.0)
(0.2,-0.04)
(0.3,0.08)
(0.4,-0.02)
(0.5,0.12)
(0.6,-0.09)
(0.7,-0.03)
(0.8,0.0)
(0.9,0.0)
(1.0,0.0)
(1.1,0.0)
(1.2,0.0)
(1.3,0.0)
(1.4,0.0)
(1.5,0.0)
(1.6,0.0)
(1.7,0.0)
(1.8,0.0)
(1.9,0.0)
(2.0,0.0)
\rotate{90}
\pscurve[liftpen=2]%
(0.0,0.0)
(0.1,0.0)
(0.2,-0.04)
(0.3,0.08)
(0.4,-0.02)
(0.5,0.12)
(0.6,-0.09)
(0.7,-0.03)
(0.8,0.0)
(0.9,0.0)
(1.0,0.0)
(1.1,0.0)
(1.2,0.0)
(1.3,0.0)
(1.4,0.0)
(1.5,0.0)
(1.6,0.0)
(1.7,0.0)
(1.8,0.0)
(1.9,0.0)
(2.0,0.0)
}
\rput[c](-1.1,0.2){\textbf{u}}
\rput[c](0.6,1){\textbf{t}}
\rput[c](0.6,-1){\textbf{f}}
\pscircle(5.0,0){2}
\pscircle(5.0,0){1}
\pscustom{%
\psline(6.0,0)(7.0,0)
\translate(5.0,0)
\rotate{150}
\translate(-5.0,0)
\psline[liftpen=2](6.0,0)(7.0,0)
\translate(5.0,0)
\rotate{40}
\translate(-5.0,0)
\psline[liftpen=2](6.0,0)(7.0,0)
}
\rput[c](3.5,0.2){\textbf{u}}
\rput[c](5.0,0){\textbf{i}}
\rput[c](5.6,1.4){\textbf{t}}
\rput[c](5.6,-1.4){\textbf{f}}
\end{pspicture}
}
\caption{Least vs typical intended $\sqsupseteq^\Quad$-model\label{fig-venn}}
\end{figure}

Figure \ref{fig-venn} gives a graphical representation of how the
least model of a program compares with a typical intended model.  In the least
model, no atoms are \textbf{i}, and (ideally) there is a correspondence between
the truth values of atoms and their behaviour.
However, the distinction between these categories can be subtle and
un-intuitive (hence the wiggly lines).
For example, the atom \verb@eq_diff([],0,[],0)@ may succeed, 
finitely fail or loop, depending on how \verb@sub/3@ is coded.
In a typical intended interpretation
there are atoms which are \textbf{i} (they may have any other truth value
in the least model).  This allows the distinction between the categories
to be more intuitive and allows a single interpretation to be a model
of many different programs with different behaviours for the \textbf{i}
atoms.  The set of \textbf{u} atoms in a typical intended interpretation
is a subset of the \textbf{u} atoms in the minimal model (often it is
the empty set, which corresponds to a three-valued model of
\citeN{sem3neg}).  Atoms which
are \textbf{u} in the minimal model can have any truth value in the
intended model.

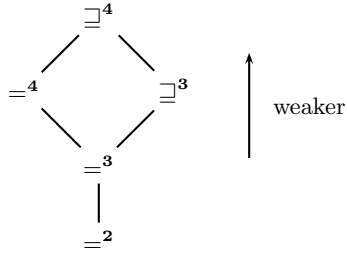
\begin{figure}[t]
\begin{center}
\begin{pspicture}(-1,0)(3,3.2) 
\psset{nodesep=3pt}

\rput[c](0,0){\rnode[c]{E2}{$=^\Bi$}}
\rput[c](0,1){\rnode[c]{E3}{$=^\Tri$}}
\rput[c](-1,2){\rnode[c]{E4}{$=^\Quad$}}
\rput[c](1,2){\rnode[c]{S3}{$\sqsupseteq^\Tri$}}
\rput[c](0,3){\rnode[c]{S4}{$\sqsupseteq^\Quad$}}

\ncline{S3}{E3}
\ncline{S3}{S4}
\ncline{E3}{E4}
\ncline{S4}{E4}
\ncline{E2}{E3}

\rput[c](2.8,1.8){weaker}
\psline[arrows=->](2,1.1)(2,2.5)

\end{pspicture}
\end{center}
\caption{Relationship between model definitions\label{fig-model-flex}}
\end{figure}

Figure~\ref{fig-model-flex} shows the relationship between the five
different definitions of a model we have considered.  Any interpretation
which is a model according to one definition is also a model according
to all definitions
above.
Weaker definitions of models
allow more flexibility in how we think of our programs, yet still
guarantee partial correctness.

\section{A ``model meet'' property}
\label{sec-model-intersect}

With the classical logic approach for definite clause programs, we have a
useful model intersection property: if $M$ and $N$ are (the set of true
atoms in) models then $M \cap N$ is (the set of true atoms in) a model.
Proposition 1 of \citeN{sem3neg} generalises this result using the
truth ordering for three-valued interpretations, and Proposition 2 of
\citeN{sem3neg} gives a similar result which mixes the truth and
information orderings.  However, none of these results hold for logic
programs with negation.  Here we give a new analogous result, using the
information ordering, which holds even when negation is present.
This will be utilised in our discussion of modes in Section\ \ref{sec-modes}.

\begin{proposition} \rm
\label{prop-mod-meet}
If $M$ and $N$ are $\sqsupseteq^\Quad$-models of program $P$ then $M \sqcap N$ is a
$\sqsupseteq^\Quad$-model of $P$.
\end{proposition}
\begin{proof}
Assume $M$ and $N$ are $\sqsupseteq^\Quad$-models of $P$.
By Proposition~\ref{prop-fourval-sqsubset},
$\Phi_P(M) \sqsubseteq M$ and $\Phi_P(N) \sqsubseteq N$, since $M$ and $N$
are models.  
By monotonicity, $\Phi_P(M \sqcap N) \sqsubseteq \Phi_P(M) \sqsubseteq M$,
and $\Phi_P(M \sqcap N) \sqsubseteq \Phi_P(N) \sqsubseteq N$.
It follows that
$\Phi_P(M \sqcap N) \sqsubseteq M \sqcap N$, 
so by Proposition~\ref{prop-fourval-sqsubset}, 
$M \sqcap N$ is a model of $P$.
\end{proof}
For example, with the models of $P3$ in Table~\ref{fig-sub-models},
$I3' \sqcap I3'' = I3$.
The corresponding result does not hold for $=^\Quad$-models. 
Consider the following program:
\begin{verbatim}
     p :- p.
     q :- q.
     r :- p ; q ; s.
     s :- p ; q ; not r.
\end{verbatim}
Let $M$ be the interpretation
which maps (\texttt{p},\texttt{q},\texttt{r},\texttt{s}) to
(\textbf{t},\textbf{f},\textbf{t},\textbf{t}), respectively, and let $N$
be the interpretation (\textbf{f},\textbf{t},\textbf{t},\textbf{t}).
Both $M$ and $N$ are $=^\Quad$-models.  The meet, $M \sqcap N$, is
(\textbf{u},\textbf{u},\textbf{t},\textbf{t}) but $\Phi_P$ applied to
this interpretation is (\textbf{u},\textbf{u},\textbf{t},\textbf{u}).
So $M \sqcap N$ is a $\sqsupseteq^\Quad$-model but not a $=^\Quad$-model.
(This example also shows that $\Phi_P$, while monotone, is not
in general an increasing function.)

\section{Program analysis}
\label{sec-proganalysis}

This section and the three that follow it present several applications 
of the four-valued semantics we have introduced.
We hope to convince the reader that there are numerous situations
in which four-valued logic is the natural setting for reasoning about
logic programs and their behaviour, and that $\sqsupseteq^\Quad$-models in
particular can play an important role.

Four-valued logic provides a convenient setting for static analysis
of logic programs.
The reason is that program analysis almost always is concerned with
runtime properties that are undecidable, so some sort of
approximation is needed, to guarantee finiteness of analysis.
As an example, we show how the program analysis framework proposed by
\citeN{Mar-Son:jlp92} generates four-valued interpretations of the kind
we have discussed.

Many program analyses for logic programs try to detect how logic
variables are being used or instantiated.
The well-known $T_P$ function and Fitting's $\Phi_P$ function yield 
ground atomic formulas, and so semantic definitions based on these
functions are not ideal as a basis for static analysis which intends 
to express what happens to variables at runtime.
The $s$-semantics \cite{Fal:ssem,Bossi:JLP94} 
is a non-ground version of the $T_P$ semantics.
The $s$-semantics of a program $P$ is a set $S_P$ of possibly
non-ground atoms, with the property that
(a) the ground instances of the atoms in $S_P$ give precisely the least 
Herbrand model of $P$, and 
(b) the computed answer substitutions~\cite{Llo84} for a query $Q$ can 
be obtained by solving $Q$ using the (potentially infinite) set $S_P$.
Letting $\cH$ denote the set of atomic formulas, the $s$-semantics
of $P$ is defined as the least fixed point of an ``immediate
consequences'' operator $T^v_P : \pset{\cH} \rightarrow \pset{\cH}$.
More precisely,
\[
  T^v_P(I) =
      \sset{h\theta} {C\equiv h\ \mathtt{:-}\ b_1,\ldots,b_n \in P,
            \tuple{a_1,\ldots,a_n} \ri_C I, \\
            \theta =
            \ident{mgu}(\tuple{b_1,\ldots,b_n},\tuple{a_1,\ldots,a_n})}
\]
where $\tuple{a_1,\ldots,a_n} \ri_C I$ expresses that $a_1,\ldots,a_n$
are variants of elements of $I$ renamed apart from $C$ and from
each other, and $\ident{mgu}$ gives the most general unifier of two
(sequences of) expressions.
As an example, for the append program,
\begin{verbatim}
     append([],Ys,Ys).
     append([X|Xs],Ys,[X|Zs]) :- append(Xs,Ys,Zs).
\end{verbatim}
the least fixed point of $T^v_P$ is
\[
  \{\mathtt{append}\mbox{\texttt{([}}x_1,\ldots,x_n\mbox{\texttt{]}},\ident{ys},\mbox{\texttt{[}}x_1,\ldots,x_n\mbox{\texttt{|}}\ident{ys}\mbox{\texttt{])}} \mid
       {n \geq 0} \}
\]
\citeN{Cod-Son:JonesFest02} give an account of the role of various
(goal-directed as well as goal-independent) semantics,
including the s-semantics, for the analysis of logic programs.

\citeN{Mar-Son:jlp92} use Fitting's three-valued semantics as a basis
for defining static analyses which over-estimate both a given
program's success set and its finite failure set.
Fitting's $\Phi_P$ operator generates pairs $(S,F)$ of
sets of ground atoms, with the reading that every atom in $S$ succeeds
and every atom in $F$ finitely fails, and it only allows for pairs
that satisfy $S \cup F = \emptyset$.
That is, there are three cases for an atom: 
It can be contained in $S$ (have the value $\mathbf{t}$), it can be
contained in $F$ (have the value $\mathbf{f}$), or it can be
absent from both (have the value $\mathbf{u}$).
In the ``approximate'' version of \citeN{Mar-Son:jlp92},
$S$ and $F$ are allowed to share atoms.
In the parlance of the present paper, 
these atoms are given the value $\mathbf{i}$, though in the context of
analysis it means ``don't know'' rather than ``don't care''.

There are different ways in which we can guarantee finiteness of the
static analysis.
One of the approaches used by \citeN{Mar-Son:jlp92} is to never
generate terms beyond a fixed, finite depth.
We define the \emph{depth} of a variable or a constant to be 1,
and for other terms we define depth inductively:
\[
  \ident{depth}(t) =
    1 + \ident{max}\ \{\ident{depth}(u) \mid u \mbox{ is a proper subterm of } t\}.
\]
The idea in \emph{depth-$k$ analysis} is that a
term with a depth greater than $k$ can be approximated by replacing
each subterm at that depth by a fresh variable.
The pruned term approximates the original in the sense that the
original is one of the pruned term's (possibly numerous) instances.

As an example, the following term, representing the
infinite list of positive odd integers
    \verb![s(0),s(s(s(0))),s(s(s(s(s(0))))),...]!
(using successor notation for non-negative integers)
has the (finite) depth-10 approximation
\begin{verbatim}
     [s(0),s(s(s(0))),s(s(s(s(s(0))))),s(s(s(s(s(s(s(0))))))),
                                          s(s(s(s(s(s(s(s(_))))))))]
\end{verbatim}
This is an imprecise approximation, as for example \verb!s(s(s(s(s(s(s(s(0))))))))!
is an instance,
but it is sufficiently precise to capture aspects of the original list,
such as the fact that it does not contain \verb!s(s(s(s(0))))!, say.

The bottom-up analysis framework proposed by \citeN{Mar-Son:jlp92} 
will mimic Fitting's $\Phi_P$ operator, except that
it generates possibly non-ground atoms, with the understanding
that a non-ground atom represents the set of all its ground instances.
This opens up the possibility that some ground atoms (instances of
non-ground atoms) be classified as both $\mathbf{t}$ and $\mathbf{f}$.
For example, for the program
\begin{verbatim}
     odd(s(0)).
     odd(s(s(N)) :- odd(N).
     p :- p.
\end{verbatim}
it will produce a depth-10 approximation
\[
 \big( 
  \{\mathtt{odd}(\mathtt{s}^j(0)) \mid j \in \{1,3,5,7\} \lor j > 8\},
  \{\mathtt{odd}(\mathtt{s}^j(0)) \mid j \in \{0,2,4,6\} \lor j > 7\}
 \big)
\]
Said differently, it creates a four-valued interpretation in which,
for example,
$\mathtt{p}$ is mapped to $\mathbf{u}$, 
$\mathtt{odd(0)}$ is mapped to $\mathbf{f}$, 
$\mathtt{odd(s(0))}$ is mapped to $\mathbf{t}$, 
and atoms $\mathtt{odd}(\mathtt{s}^j(0))$ are mapped to 
$\mathbf{i}$ for all $j>8$.

The ``approximate'' semantics that underpins the bottom-up analysis 
framework is perfectly aligned with the semantics proposed in 
Section~\ref{sec-generalised}.
A bottom-up analysis of program $P$ is expressed as an interpretation $I =
\lfp(\Phi'_P)$ for some sound approximation $\Phi'_P$ of $\Phi_P$. In the
case where all atoms are ground, $I$ may also be a fixed point of $\Phi_P$,
that is, $I = \Phi_P(I)$ and $I$ is a $=^\Quad$-model.  The definition of
``sound approximation'' (which ensures that the approximation must be
at least as high in the information order as the components it approximates)
implies $\Phi'_P(I) \sqsupseteq \Phi_P(I)$.  In general, since $I =
\Phi'_P(I)$, we have $I \sqsupseteq \Phi_P(I)$ so $I$ is a
$\sqsupseteq^\Quad$-model.

\section{Types and modes}
\label{sec-modes}

We now briefly review the motivation for type and mode systems in logic
programming and show how $\sqsupseteq^\Quad$-models can play a role
in designing and understanding mode systems.
Along the way we propose an expressive mode system, 
and a language of mode annotations to support it.

The lack of a type discipline or similar restrictions on what 
constitutes acceptable Prolog programs means that it is easy for 
programmers to make simple mistakes which
are not immediately detected by the Prolog system.  
A typical symptom
is that the program fails unexpectedly, leading to rather tedious analysis
of the complex execution in order to uncover the mistake.  
One approach to avoid
some runtime error diagnosis is to impose additional discipline on the
programmer, generally restricting programming style somewhat, 
to allow the system to statically classify certain programs as incorrect.
Various systems of ``types'' and ``modes'' have been proposed for this.
An added benefit of some such systems is that they help
make implementations more efficient.  
Here we discuss systems of this kind
at a very high level and argue that four-valued interpretations
potentially have a role in this area, particularly in mode systems such
as that of Mercury \cite{mercury}.

Type systems typically assign a type (say, Boolean, integer, list of
integers) to each argument of each predicate.  This allows each variable
occurrence in a clause to also be assigned a type.  One common error
is that two occurrences of the same variable have different types.
For example, consider a predicate \texttt{head} which is intended to
return the head of a list of integers but is incorrectly defined as:
\verb@head([_|Y],Y)@.  The first occurrence of \texttt{Y} is associated
with the type ``list of integers'' 
and the other is associated with type ``integer''.
If \texttt{head} is called with both arguments instantiated to the
expected types, it must fail.  But \texttt{head} can succeed if it
is called in different ways.  For example, with only the first argument
instantiated it will succeed, albeit with the wrong type for the second
argument (and this in turn may cause a wrong result or failure of a 
computation which calls \texttt{head}).

Type systems can be refined by considering the ``mode'' in which
predicates are called, or dependencies between the types of different
arguments.  This can allow additional classes of errors to be detected.
For example, we can say the first argument of \texttt{head} is expected
to be ``input'' and the second argument can be ``output''.  Alternatively
(but with similar effect), we could say if the first argument is a list
of integers, the second should be an integer.  
To see that mode information transcends type information, consider the
(incorrect) definition \verb@head([_|Y],X)@.
Here there is a consistent assignment of types to variables, 
but it does not satisfy the stipulated mode/type-dependency constraint.
One high level property of several mode systems is that if input
arguments are well typed then output arguments will be well typed for
any successful call.  
In fact, a stronger property is desirable: the whole successful derivation
should be well typed (otherwise we have a very dubious proof).  Typically,
well typed inputs in a clause head imply well typed inputs in the body,
which implies well typed outputs in body, which implies well typed
outputs in the head.  This idea is present in the ``directional types''
concept \cite{Aiken-sas94,boye95}, the mode system of Mercury
\cite{mercury}, and the view of modes proposed in \citeN{modes}.
Here we show the relevance of four-valued interpretations to this idea,
ignoring the details of what constitutes a type (which differs in the
different proposals) and what additional constraints are imposed 
(neither Mercury nor directional types support cyclic dataflow, 
and Mercury has additional interactions between types, 
modes and determinism).

We will present a mode system inspired by that of Mercury.  Mercury
allows types to be defined using \texttt{type} declarations and
declared for predicate arguments using \texttt{pred} declarations; we
adopt this verbatim.  Mercury modes are declared using \texttt{mode}
declarations, which also declare determinism (the range of possible
numbers of solutions).  We propose similar mode declarations which allow
additional refinements.  Determinism information could also be added
but we ignore this aspect here.  Similarly, we ignore issues surrounding
negation.  Mercury also supports an
\texttt{error/1} primitive which results in abnormal termination if
called.  It allows more precise static analysis of modes 
(and of determinism) and we also adopt it.

Type and mode declarations document some aspects of how predicates are
intended to be used and how they are intended to behave.  We define
interpretations which are consistent with these documented intentions.
We assume there is a notion of well typedness for each argument of each
predicate in program $P$.

\begin{definition}[Mode, mode group, mode interpretation and moded
program] \rm
A \emph{mode} for predicate $P$ is an assignment of ``input'' or 
``output'' to each of $P$'s argument positions.
A \emph{mode interpretation} of $P$ with a given mode $m$, 
$\ident{MI}(P, m)$, is a four-valued interpretation $M$
such that the value of (ground) atom $A$ in $M$ is
\begin{itemize}
\item
\textbf{u}, if the predicate is \texttt{error/1}, and otherwise:
\item
\textbf{t}, if all arguments are well-typed,
\item
\textbf{i}, if some input argument (according to $m$) is ill-typed, and
\item
\textbf{f}, if all input arguments are well typed but some output
argument is ill-typed.
\end{itemize}
A \emph{mode group} is a set of modes for a predicate.  A \emph{mode
interpretation} of $P$ with a mode group $\{m_1 \ldots m_k\}$ is
$\bigsqcap_{1 \leq i \leq k} \ident{MI}(P, m_i)$.  A \emph{moded program}
is a program with a mode group defined for each predicate.  Mode
interpretations for moded programs are defined in the obvious way.
\end{definition}

Note that the assignments \textbf{u} and \textbf{t} are independent 
of the mode(s).

For a mode group, an atom is \textbf{i} where there is no mode in
the group for which all input are well typed.  Changing the mode(s) of a
predicate so it can be used in more flexible ways corresponds to changing
the truth value of some atoms from \textbf{i} to \textbf{f}.  This makes
the mode interpretation more precise (lower in the information order).
Asymmetry between \textbf{t} and \textbf{f} arises because mode
analysis must ``assume the worst'' with respect what can succeed, local
variables in clauses are existentially quantified in the body, and
negated literals do not bind variables which appear in the rest of the
clause body. 
Figure \ref{fig-nrev} illustrates the syntax we use for defining the
modes of a predicate---mode groups are formed using the keyword ``and''.
The Mercury equivalent is to use two separate mode declarations.

\begin{figure}
\begin{verbatim}

                :- pred rev(list(T), list(T)).
                :- mode rev(in, out) and (out, in).
                rev([], []).
                rev([H|T], R) :- rev(T, L), append(L, [H], R).
\end{verbatim}
\caption{Naive reverse with a group of two modes\label{fig-nrev}}
\end{figure}

Mercury uses the notion of an \emph{implied mode}---a mode $m$
implies all modes whose output arguments are a subset of those in $m$.
For example, the mode \texttt{(in,in)} is implied by either of the two
declared modes for \texttt{rev/2}.  
The next proposition says that mode 
interpretations are invariant under addition of implied modes.

\begin{proposition}
\label{prop-imp-mode}
The mode interpretation of predicate $P$ with modes $\{m_1 \ldots m_k\}$
is the same as that for $P$ with modes $\{m_1 \ldots m_k, m'\}$ if
the outputs of $m'$ are a subset of the outputs of some $m_j$, $1 \leq j
\leq k$.
\end{proposition}
\begin{proof}
For each $j$ we have that $\ident{MI}(P, m') \sqsupseteq \ident{MI}(P, m_j)$, 
since each input argument of $m_j$ is an input argument of $m'$.
So $\ident{MI}(P, m') \sqsupseteq \bigsqcap_{1 \leq j \leq k} \ident{MI}(P, m_j)$
and thus
$\ident{MI}(P, m') \sqcap \bigsqcap_{1 \leq j \leq k} \ident{MI}(P, m_j) =
\bigsqcap_{1 \leq j \leq k} \ident{MI}(P, m_j)$.
\end{proof}

\noindent
If a mode interpretation of a moded program $P$ is a
$\sqsupseteq^\Quad$-model
this gives us the high level property discussed earlier:
\begin{proposition} \rm
\label{lem-mode-head}
If a mode interpretation $M$ of a moded program $P$ is a
$\sqsupseteq^\Quad$-model and $A$ is a
successful atom which, for some mode of the predicate, has all input
arguments well typed, then $A$ has all arguments well typed.
\end{proposition}
\begin{proof}
By Corollary \ref{thm-sound}, since $M$ is a $\sqsupseteq^\Quad$-model
and $A$ succeeds,
$A$ must be \textbf{t} or \textbf{i} in $M$.  By the definition of mode
interpretations, since $A$ is not \textbf{f} and all input arguments
are well typed for some mode, all output arguments must be well typed
as well.
\end{proof}
For the stronger property to hold (the whole derivation being well-typed),
the mode interpretation being a $\sqsupseteq^\Quad$-model is not
sufficient.  A definition can have a \textbf{t} head and a body which is a
disjunction of a \textbf{t} atom which loops and an \textbf{i} atom which
succeeds: although the interpretation is a $\sqsupseteq^\Quad$-model, the
only successful derivation uses the inadmissible disjunct.  To prevent
such cases we impose an extra condition on each disjunct in the body of
a definition (or each clause in a Prolog program) rather than the body
of the definition as a whole:

\begin{definition}[Well-moded] \rm
A moded program $P$ is \emph{well-moded} with respect to
mode interpretation $M$ if $M$
is a $\sqsupseteq^\Quad$-model of $P$ and for each head grounding of a
definition $(H,\exists W[C_1 \lor \cdots \lor C_k])$ where $M(H) =
\mathbf{t}$, $M(\exists W ~ C_j) \neq \mathbf{i}$, $1 \leq j \leq k$.
\end{definition}

\noindent
In practice, it seems that the additional constraint rarely makes a
difference.  In the examples we discuss below, whenever the interpretation
is a $\sqsupseteq^\Quad$-model, the program is well-moded with respect
to the interpretation.

\begin{lemma} \rm
\label{lem-mode-body}
If program $P$ is well-moded, with mode interpretation $M$ and
atom $A$, with $M(A) = \mathbf{t}$, succeeds, then
$A$ is well typed and there is a ground
predicate definition instance $(A, \exists W[C_1 \lor \cdots \lor C_k])$
such that all literals in some $C_j$ succeed and are
assigned \textbf{t} and all positive literals in $C_j$ are well typed.
\end{lemma}
\begin{proof}
Since $M$ is a mode interpretation and $M(A) = \mathbf{t}$,
all arguments are well typed.
A successful disjunct $C_j$ must exist; it cannot be \textbf{i} since
$P$ is well-moded, so it must be \textbf{t} (only \textbf{i} and
\textbf{t} disjuncts can succeed, by Corollary
\ref{thm-sound}).  Similarly, no literal in $C_j$ can be \textbf{u}, so
all literals in $C_j$ must be \textbf{t}, thus each positive literal in
$C_j$ must be well typed.
\end{proof}

\begin{theorem} \rm
\label{thm-mode-deriv}
If $P$ is a well-moded program
and $A$ is a successful atom which is \textbf{t} in the mode
interpretation of $P$, then there is successful derivation of $A$ in
which all successful atoms are well typed.
\end{theorem}
\begin{proof}
By induction on the depth of the proof and Lemma~\ref{lem-mode-body}.
\end{proof}

\noindent
Checking that a mode interpretation is a $\sqsupseteq^\Quad$-model
(and the additional constraint holds)
requires the kind of analysis used in other forms of mode checking.
For example, consider again Figure~\ref{fig-nrev}.
Assume the recursive clause for \texttt{rev/2}
uses mode \texttt{(in,out)} and assume \texttt{append/3} has
mode \texttt{(in,in,out)}.  Mode analysis intuitively reasons that if
\texttt{rev/2} is called with the first argument well typed (\verb@[H|T]@
is a list), the recursive call will be called with its input argument
well typed (\texttt{T} is a list), thus in any successful call its
output argument will be well typed (\texttt{L} is a list), the input
arguments to \texttt{append/3} will be well typed so its argument will
be well typed (\texttt{R} is a list), so the head will be well typed.
In other words, if we assume the head is \textbf{t} or \textbf{f},
we can find an instance of the body which is \textbf{t} and the head
must be \textbf{t}.  Thus there are no head clause instances of the
form \textbf{t}\texttt{:-}\textbf{i}, \textbf{f}\texttt{:-}\textbf{i},
\textbf{f}\texttt{:-}\textbf{t} or \textbf{t}\texttt{:-}\textbf{f}
(and the head cannot be \textbf{u} since the predicate is not 
\texttt{error/1}), 
so the mode interpretation is a $\sqsupseteq^\Quad$-model.

\begin{figure}
\begin{verbatim}

     % Extracts head of list.  Has exactly one solution for all
     % (admissible) calls.  nonempty_head([],_) is inadmissible.
     :- pred nonempty_head(list(T), T).
     :- mode nonempty_head(in, out).
     nonempty_head([H|_], H).

     % Extracts head of list.  Has exactly one solution for all
     % (normally terminating) calls.  checked_head([],_) throws an error.
     :- pred checked_head(list(T), T).
     :- mode checked_head(in, out).
     checked_head([], _) :- error("head of empty list").
     checked_head([H|_], H).
\end{verbatim}
\caption{Two versions of head for non-empty lists\label{fig-head}}
\end{figure}

Sometimes the most natural intended interpretation has certain atoms
assigned \textbf{i}, but static mode analysis is unable to conclude
such atoms are never called.  In Mercury this issue arises more with
determinism, when analysis is unable to conclude that a given
atom always succeeds.
Using \texttt{error/1} allows us to make static analysis more flexible
by making the code more verbose and introducing some runtime checking.
As an example of this, consider extracting the head of a list.  
There are situations
where we expect this operation to be called on non-empty lists only
(so the computation has exactly one solution, or is ``det'' in Mercury
terminology).  Figure \ref{fig-head} gives two codings.  The first
cannot easily be checked statically.  Our proposed mode system has
no way to declare well-typed atoms as inadmissible, so clauses of the
form \textbf{t}\texttt{:-}\textbf{i} in our intended interpretation may
be possible.  Similarly, Mercury cannot determine the code is ``det''.
The second is acceptable for Mercury and also safe for our mode system.
By changing the intended interpretation of \verb@checked_head([],_)@ from
\textbf{i} to \textbf{u}, the mode interpretation (where it is assigned
\textbf{t}) becomes a safe approximation to the intended interpretation.
Thus, having mode interpretations that can distinguish \textbf{i} 
from \textbf{u} can be helpful.

There is one more feature of the mode system we propose---allowing more
than one mode group per predicate.  This feature is not supported in any
other mode systems.  Separate mode groups are declared using the keyword
``also''.  In the following example each group has a single mode, but in
general we can use a mixture of ``and'' and ``also'', with the former
binding more tightly.

\begin{verbatim}
     :- mode rev(in, out) also (out, in).
\end{verbatim}

\noindent
For the external view of a predicate, for example, how modes approximate
the behaviour of a non-recursive call to a predicate, ``also'' is
treated identically to ``and'' (the meet of the mode interpretations is
used).  However, for the internal view of a predicate and how its
definition is checked for well-modedness, we impose a
stronger constraint.  The definition must be well-moded
with respect to \emph{each} interpretation corresponding to a mode group.
This implies it is also well moded with respect to the meet (we give the
case of two mode groups; the generalisation to $N$ mode groups is
straightforward):

\begin{proposition}
\label{prop-meet-well-moded}
If predicate $P$ is well-moded with respect to mode interpretations
$\ident{MI}_1$ and $\ident{MI}_2$ then it is well-moded with respect to
$M$, where $M = \ident{MI}_1 \sqcap \ident{MI}_2$.
\end{proposition}
\begin{proof}
$M$ is a $\sqsupseteq^\Quad$-model, by Proposition \ref{prop-mod-meet}.
Consider a head grounding of the definition of $P$, 
$(H,\exists W[C_1 \lor \cdots \lor C_k])$.
If $M(H) = \mathbf{t}$ then $\ident{MI}_1(H) = \mathbf{t}$ or
$\ident{MI}_2(H) = \mathbf{t}$, so
$\ident{MI}_1(C_j) \neq \mathbf{i}$ or
$\ident{MI}_2(C_j) \neq \mathbf{i}$, so
$M(C_j) \neq \mathbf{i}$, for $1 \leq j \leq k$.
\end{proof}
For sets of mutually recursive predicates there must be some set of mode
interpretations $S$, the predicates must be well-moded with respect to
each element of $S$, and $S$ must have a mode interpretation for each
mode group in each of the predicates (each mode group of a predicate
must be ``covered'' by at least one element of $S$).

\begin{figure}
\begin{verbatim}
            :- pred rev_ra(list(T), list(T)).
            :- mode rev_ra(in, out) also (out, in).
            rev_ra([], []).
            rev_ra([H|T], R) :- rev_rb(L, T), append(L, [H], R).

            :- pred rev_rb(list(T), list(T)).
            :- mode rev_rb(in, out) also (out, in).
            rev_rb([], []).
            rev_rb([H|T], R) :- rev_ra(L, T), append(L, [H], R).
\end{verbatim}
\caption{Mutually recursive reverse with complementary modes\label{fig-rrev}}
\end{figure}

Consider the naive reverse example with the mode declaration above
and assume \texttt{append/3} has modes \texttt{(in,in,out)} and
\texttt{(out,in,in)}.  There are two $\sqsupseteq^\Quad$-models
corresponding to the mode interpretations for modes
\texttt{(in,out)} and \texttt{(out,in)}, respectively, and a third
$\sqsupseteq^\Quad$-model which is the meet.  However, if we swap the
arguments in the recursive call to \texttt{rev/2}, the meet is still
a $\sqsupseteq^\Quad$-model but the other two mode interpretations are
\emph{not} $\sqsupseteq^\Quad$-models.  Calling \texttt{rev/1} in mode
\texttt{(in,out)} requires a recursive call in mode \texttt{(out,in)}
and vice versa, so one mode alone is not sufficient and mode checking
with the ``also'' mode declaration would fail.  Mode declarations with
``also'' are stronger than those with ``and''; they tell us more about the
\emph{set} of $\sqsupseteq^\Quad$-models.  The additional expressiveness
can be used to detect more errors (for example, if the arguments were
swapped accidentally and the stronger mode declaration was used).

The version of reverse with the arguments swapped can be specialised
to two mutually recursive predicates, shown in Figure \ref{fig-rrev}.
This has three $\sqsupseteq^\Quad$-models: one with mode \texttt{(in,out)}
for \texttt{rev\_ra/2} and mode \texttt{(out,in)} for \texttt{rev\_rb/2},
another with mode \texttt{(out,in)} for \texttt{rev\_ra/2} and mode
\texttt{(in,out)} for \texttt{rev\_rb/2}, and the third is the meet.
The program is well-moded with respect to all three and each mode group
of each predicate is covered by one of these models.

\begin{figure}
\begin{verbatim}
    :- type b ---> t ; f.          % Boolean
    :- type k3 ---> t3 ; f3 ; i3.  % Kleene
    
    % 'and' where third truth value means maybe true, maybe false
    :- pred and3(k3, b, b).
    :- mode and3(in, in, out).
    and3(i3, _, f).
    and3(i3, t, t).
    and3(f3, _, f).
    and3(t3, B, B).
    
    % 'and3' of each value in a list
    :- pred fold_and3(list(k3), b).
    :- mode fold_and3 (in, out) and (in, in).  % latter is redundant
    fold_and3([], t).
    fold_and3([f3|_], f).
    fold_and3([B3|B3s], R) :- fold_and3(B3s, R1), and3(B3, R1, R).
    
    :- pred fold_and3a(list(k3), b).
    :- mode fold_and3a (in, out) also (in, in).
    fold_and3a([], t).
    fold_and3a([f3|_], f).
    fold_and3a([i3|_], f).
    fold_and3a([i3|B3s], t) :- fold_and3a(B3s, t).
    fold_and3a([t3|B3s], R) :- fold_and3a(B3s, R).
\end{verbatim}
\caption{Illustration of ``and'' versus ``also'' in modes\label{fig-and3}}
\end{figure}

Figure \ref{fig-and3} gives another example of the additional expressive
power.  The mode declared for \verb@fold_and3/2@ is redundant: it has
\texttt{(in,out)} and the implied mode \texttt{(in,in)}.  However, even
though \texttt{(in,in)} is weaker in some sense, and its corresponding
mode interpretation is strictly higher in the information ordering, it
is not a $\sqsupseteq^\Quad$-model.  Calling \verb@fold_and3/2@ in mode
\texttt{(in,in)} requires a recursive call in mode \texttt{(in,out)}.
However, for \verb@fold_and3a/2@, which computes the same thing, the
code is well-moded with respect to each of the mode interpretations, as
expressed by the ``also''.
The mode \texttt{(in,in)} does not rely on mode \texttt{(in,out)} and
considerably better efficiency can be achieved,
because it can be statically determined (by the Mercury compiler, for
example) that no choice points are needed.

Precise analysis of declared types, modes, determinism, and so on,
is useful for uncovering program errors statically and increasing 
efficiency of implementations.  
Such analysis distinguishes computations which
(might) succeed from those which (must) fail.  Most proposals also support
methods to restrict the ways in which predicates should be used, 
for example, the input arguments should be well typed.  
The more advanced proposals also support forms of abnormal termination,
such as \texttt{error/1}.
The four-valued domain we use for the semantics of logic programs 
seems particularly well suited to this kind of analysis.
In particular, we have demonstrated how type and mode declarations can
be used to define four-valued interpretations and how
$\sqsupseteq^\Quad$-models are an important device for checking
correctness of these declarations.

\section{Formal Specifications}
\label{sec-specification}

\begin{figure}
$\forall s \forall s' ~subset(s, s') \leftrightarrow \forall e~ (member(e, s) \rightarrow
member(e, s'))$

\begin{verbatim}
              subset([], _).
              subset([E|SS], S) :- member(E, S), subset(SS, S).

              member(E, [E|_]).
              member(E, [_|S]) :- member(E, S).

              list([]).
              list([_|S]) :- list(S).
\end{verbatim}
\caption{First-order logic specification and Prolog definition of
\texttt{subset/2} \label{fig-subset}}
\end{figure}

In the early days of logic programming there was considerable interest
in the relationship between specifications (particularly formal
specifications written in classical first order logic) and logic programs
\cite{ClarkSickel,Hog81,SatoTamaki,Kow85}.
This work generally overlooked what we here have called inadmissibility.  
For example, Figure \ref{fig-subset}
shows a specification and Prolog implementation of the \texttt{subset/2}
predicate given by \citeN{Kow85}, where sets are represented as lists and
$member$ is the Prolog list membership predicate.  \citeN{Kow85} shows
that the implementation is a logical consequence of the specification.
That is to say, the program $P$ which defines \verb!subset!
is sound with respect to the specification $S$:
for all queries $Q$, if $P \models Q$ then $S \models Q$.
However, \texttt{subset(true,42)} is true according to the specification,
which is counter-intuitive, to say the least.  If the specification
is modified to restrict both arguments to be lists, the program is no
longer a logical consequence (the program has \verb@subset([],42)@ as
a consequence but \verb@subset([],42)@ is no longer a consequence of
the specification).  
When negation is also considered, or even the fact that logic programs 
implicitly define falsehood of some atoms, early approaches relating
formal specifications and logic programs based on classical logic seem 
unworkable.  

With our approach it is natural to identify specifications with
four-valued interpretations.  Our ``intended interpretations'' are
essentially specifications, albeit informal ones which exist only in the
mind of programmers.  However, we can also design formal specification
languages where the meaning of a specification is a single four-valued
interpretation.  We propose such a language now.  Although we can never be
sure that a formal specification accurately captures our intentions
(as Kowalski's specification above shows), and fully automated
verification is bound to be intractable in general,
cross-checking between a specification and code can give us
additional confidence in
the correctness of our code.  In the design of our specification language
we aim to utilise classical logic as far as possible, while allowing
the flexibility of all four values.  Underspecification is supported by
declaring preconditions as well as postconditions.

\begin{definition}[Specification]
A specification is a well formed formula (wff) $\Delta$, a set of distinct
atoms $A_i$ in most general form, a precondition wff $\alpha_i$ for each
$A_i$ and a postcondition wff $\omega_i$ for each $A_i$.
\end{definition}

For example, a precondition and postcondition of \texttt{subset/2}
could be defined using syntax exemplified below 
(which could be supported by
just declaring the three keywords as operators in NU-Prolog or Mercury).
In addition, $\Delta$ would define the predicates \texttt{member/2} and
\texttt{list/1} using a syntax close to traditional first order logic,
or Prolog syntax could be used as shorthand for the Clark completion,
for example.

\begin{verbatim}
     predicate subset(SS, S)
     precondition list(S), list(SS)
     postcondition all [E] (member(E, SS) => member(E, S)).
\end{verbatim}

\begin{definition}[Meaning of a specification]
The meaning of a specification is a four-valued interpretation for the
$A_i$ predicates, such that each ground atom $A_i \theta$ is
\begin{itemize}
\item \textbf{i}, if precondition $\alpha_i \theta$ is
false in any classical model of $\Delta$, otherwise
\item \textbf{t}, if postcondition $\omega_i \theta$ is
true in all classical models of $\Delta$,
\item \textbf{f}, if postcondition $\omega_i \theta$ is
false in all classical models of $\Delta$, and
\item \textbf{u} otherwise.
\end{itemize}	
\end{definition}
With this, a \texttt{subset/2} atom which has some non-list argument
is given the value \textbf{i}, and no \texttt{subset/2}
atom gets the value \textbf{u}. 
The other \verb!subset! atoms are partitioned into \textbf{t} and
\textbf{f} in the intuitive way.  
Thus the counter-intuitive consequences
of Kowalski's specification are mapped to
\textbf{i} rather than \textbf{t}.

Kowalski and others attempted to relate the meanings of formal
specifications and programs via the truth ordering.  In our approach
we relate them via the information ordering.  A program is correct
with respect to a specification if and only if the meaning of the
specification is greater than or equal to ($\sqsupseteq^\Quad$)
the least $\sqsupseteq^\Quad$-model of the program.  The meaning of
the specification being a $\sqsupseteq^\Quad$-model of the program
is a sufficient condition for this and Theorem\ \ref{thm-sound-ge}
gives the partial correctness results.  For example, the meaning of
the \texttt{subset/2} specification is a $\sqsupseteq^\Quad$-model of
the program.  There can be different logic programs, with different
behaviours, which are correct according to a specification---they can
be seen as refinements of the specification.
If the specification is not a $\sqsupseteq^\Quad$-model of the program,
the program may succeed or finitely fail
in ways which are inconsistent with the specification (wrong answers or
missing answers).

As we develop an implementation from an initial high level specification,
we generally move lower in the information order.  For example,
we may find it useful to strengthen the specification above so
\texttt{subset/2} can be used in more flexible ways in our system.
Moving the \texttt{list(SS)} constraint from the precondition to the
postcondition is like changing the mode declaration for \texttt{subset/2}
from \texttt{(in,in)} to \texttt{(out,in)}.  The meaning of the
stronger specification is lower in the order, with some previously
\textbf{i} atoms such as \verb@subset(abc,[])@ now being \textbf{f}.  As we
proceed from specification to code we go lower still: \verb@subset([],42)@
is \textbf{i} in both specifications but \textbf{t} in the least
$\sqsupseteq^\Quad$-model of the program.  We believe our approach of
having a complete lattice in the information order can provide a simple,
elegant and accurate view of the relationship between specifications
and programs.

Our proposed specification language is inspired in part by the VDM-SL
specification language, which has preconditions but is based on the
functional programming paradigm (functions are specified rather than
predicates).  The underlying theory is the logic of partial functions, 
LPF \cite{barringer-cheng-cbjones:1984,jones_middelburg}, 
with three truth values.  Although specifications
have preconditions, the primary use of the non-classical truth value
is to represent undefined, or \textbf{u}---the meaning of a recursively
defined function is given by the least fixed point of the definition.
There is no separate truth value to represent unspecified, or \textbf{i}.  
Once again we contend that these semantically distinct notions are 
best represented by distinct truth values.

\section{Declarative debugging}
\label{sec-debugging}

The semantics of Naish~\citeyear{sem3neg} is closely aligned with
declarative debugging (introduced in \citeN{Sha83}) and the term
``inadmissible'' comes from this area \cite{Per86}.  The Naish semantics
gives a formal basis for the three-valued approach to declarative
debugging of \citeN{ddscheme3} (using \textbf{t}, \textbf{f}, and
\textbf{i}) as applied to Prolog.  Given a goal whose behaviour
is inconsistent with the intended three-valued interpretation (it
has a wrong or missing answer), the debugger identifies some part
of the code (such as a clause instance) which demonstrates that the
intended interpretation is not a $\sqsupseteq^\Tri$-model.  As we have
demonstrated in Section~\ref{sec-generalised}, four values allow us
to express programmer intentions more precisely than three.  In this
section we sketch how four-valued interpretations can be supported by
declarative debuggers.

The declarative debugging scheme represents the computation as a tree;
sub-trees represent sub-computations.  
Each node is classified by an oracle as correct, erroneous
or inadmissible.  The debugger searches the tree for a \emph{buggy}
node, which is an erroneous node with no erroneous children.
If all children are correct it is called an \emph{e-bug}, otherwise
(it has an inadmissible child) it is called an \emph{i-bug}.
Every finite tree
with an erroneous root contains at least one buggy node and finding such
a node is the job of a declarative debugger.
To diagnose wrong answers in Prolog a proof tree \cite{Llo84} is used to
represent the computation.  Nodes containing \textbf{t}, \textbf{f} and
\textbf{i} atoms are correct, erroneous and inadmissible, respectively.
To diagnose computations that \emph{miss} answers, a different form
of tree is used, and nodes containing finitely failed \textbf{t},
\textbf{f} and \textbf{i} atoms are erroneous, correct, and inadmissible,
respectively.  \citeN{ddscheme3} also spells out how to deal with some
additional complexities which arise, such as non-ground wrong answers
and computations which return some but not all correct answers; we skip
the details here.  Buggy nodes correspond to instances of definitions of
the form
\textbf{t}\texttt{:-}\textbf{f},
\textbf{f}\texttt{:-}\textbf{t},
\textbf{t}\texttt{:-}\textbf{i} or
\textbf{f}\texttt{:-}\textbf{i}.
The first two are e-bugs (the kind diagnosed by more conventional
two-valued declarative debuggers); the last two are i-bugs.

Four-valued interpretations can be used in place of three-valued
interpretations in this scheme, as follows.  The debugging algorithm remains
unchanged; only the way the oracle classifies nodes is modified.
For wrong answer diagnosis, \textbf{u} is treated the same
as \textbf{f}---a sub-computation which succeeds contrary to our
intentions is erroneous.  For missing answer diagnosis \textbf{u} is
treated the same as \textbf{t}---a sub-computation which finitely
fails contrary to our intentions is also erroneous.  This simple
generalisation of the three-valued scheme allows us to use four-valued
interpretations and
find bugs corresponding to instances of definitions where the head is
\textbf{u} but the body is not.

For example, an atom such as \texttt{interpret("main:-main.")} may be
considered admissible, since its argument is well-formed.  However,
it is not intended to terminate and if it succeeds (or finitely fails)
we would like a tool to help debug it.  With four values, we can say
this atom is \textbf{u} and if the atom appears in the node of a proof
tree, the node would therefore be considered erroneous and amenable to
declarative debugging.  The intended interpretation is not a
$\sqsupseteq^\Quad$-model and the debugger is able to diagnose why.

Intuition may suggest the debugger would need four classes of
nodes for the four truth values.  However, the classes of nodes do
not all correspond to truth values in the intended interpretation.
They correspond to the \emph{comparison} between the truth value in the
intended interpretation and the observed behaviour (or the truth value
in the least model of the program).  Note that the observed behavior is
two-valued in these uses of declarative debugging---the computation must
succeed or finitely fail.  Inadmissible nodes correspond to a
comparison using $\sqsupset$
(which only holds when the intended value is \textbf{i}).  Correct nodes
correspond to $=$.  Erroneous nodes correspond to incomparability for
three-valued interpretations but they may also correspond to $\sqsubset$
in the four-valued case.  Thus four-valued interpretations add some
flexibility to declarative debuggers with very little additional cost.

\section{Computation and the information ordering}
\label{sec-computation-information}

The logic programming paradigm introduced the view of computation as
deduction \cite{Kow80}.  Classical logic was used and hence computation
was identified with the truth ordering.  With Prolog programs viewed
as Horn clauses, \texttt{:-} is classical $\leftarrow$, or $\ge$ in
the truth ordering.  We view the Prolog arrow as $\sqsupseteq^\Quad$,
which naturally leads to identifying computation with the information
ordering rather than the truth ordering.  In this section we sketch this
alternative view of the logic programming paradigm.  The information
ordering ordering holds if we compare
successive states of a computation using a correct program (that is, the
intended interpretation is a $\sqsupseteq^\Quad$-model).  Because $H
\sqsupseteq B$ for each head grounding, replacing a subgoal by the body
of its definition (a basic step in a logic programming computation)
gives us a new goal which is lower (or equal) in the information
ordering (see Proposition~\ref{prop-comp-order}).

This view is obscured if we view Prolog computation as SLD derivations
because SLD derivations include the ``success continuation'' of the
current sub-goal but not the ``failure continuation''---the alternatives
which would be explored on backtracking.  We view a computation state
as a disjunction of (conjunctive) goals.  This is equivalent to a
\emph{frontier} of nodes in an SLD tree rather than a single node (or a
single goal in an SLD derivation).  Free variables are those appearing
in the top-level goal; other variables are existentially quantified.
A computation step selects a node from the frontier (a disjunct), then
selects a subgoal within it (a conjunct).  For simplicity, we do not
deal with negation here.  A more detailed model of logic programming
computation in this style would also include propagation of failure from
unsatisfiable equations.

\begin{definition}[Computation state, successsor state]
A computation state $S$ is a formula of the form
$\exists {V} (D_1 \vee \ldots \vee D_m)$, with each $D_i$
a conjunction of literals
$C_{i,1} \wedge \ldots \wedge C_{i,m^i}$.
Let $(C_{i,j}, \exists {W}~ (B_1 \vee \ldots \vee B_n))$ be a
head instance of a definition, with variables in ${W}$
renamed so they are distinct from those in $S$.
Let $D'$ be
$(C_{i,1} \wedge \ldots C_{i,j-1} \wedge B_1 \wedge C_{i,j+1} \ldots
\wedge C_{i,m^i}) \vee \ldots \vee
(C_{i,1} \wedge \ldots C_{i,j-1} \wedge B_n \wedge C_{i,j+1} \ldots
\wedge C_{i,m^i})$.  Then
$S' = \exists {V} \exists {W} (D_1 \vee \ldots D_{i-1} \vee D' \vee
D_{i+1} \ldots \vee D_m)$ is a successor state of $S$.
\end{definition}

Given a top-level Prolog goal, the intended interpretation gives a
truth assignment for each ground instance.  Subsequent resolvents can
also be given a truth assignment for each ground instance of the
variables in the top level goal (with local variables considered
existentially quantified).  As the computation progresses, the
truth value assignment for each instance often remains the same,
but can become lower in the information ordering.

\pagebreak
\begin{proposition}
\label{prop-comp-order}
If $S'$ is the successor state of $S$, $\theta$ is a grounding
substitution for just the free variables in $S$ (and $S'$) and
interpretation $M$ is a $\sqsupseteq^\Quad$-model of the program,
then $M(S \theta) \sqsupseteq^\Quad M(S' \theta)$.
\end{proposition}
\begin{proof}
Since variables in ${W}$ are not in $S$, $S' \theta =
(\exists {V} (D_1 \vee \ldots D_{i-1} \vee (\exists {W}~D')  \vee
D_{i+1} \ldots \vee D_m) \theta$.
Since variables in ${W}$ are not in $C_{i,k}$ for $k \ne j$
and De Morgan's laws hold for \Quad, $D' \theta = 
(C_{i,1} \wedge \ldots C_{i,j-1} \wedge
\exists {W}~(B_1 \vee \ldots \vee B_n) \wedge C_{i,j+1} \ldots
\wedge C_{i,m^i}) \theta$.
Since $M$ is a $\sqsupseteq^\Quad$-model, $ M(C_{i,j} \theta) \sqsupseteq
M(\exists {W}~ (B_1 \vee \ldots \vee B_n) \theta)$.
The result follows from the monotonicity of $\wedge$ and $\vee$.
\end{proof}
For example, consider the goal \texttt{implies(X,f)}
(where \verb!implies! was defined in Section~\ref{sec-scene}). 
Our intended interpretation maps
\texttt{implies(f,f)} to $\mathbf{t}$ and \texttt{implies(t,f)} to
$\mathbf{f}$, but may map \texttt{implies(42,f)} to $\mathbf{i}$, if the
first argument is expected to be input.  After one step of the computation
we have the conjunction \texttt{neg(X,U), or(U,f,t)}
(where \texttt{U} is local to the computation and hence existentially
quantified).
If our intended interpretation allows any mode for \texttt{neg}, the
instance where \texttt{X} = 42 is then mapped to $\mathbf{f}$.

We believe that having a complete lattice using the information
ordering provides an important and fundamental insight into the nature
of computation.  At the top of the lattice we have an element which
corresponds to underspecification in the mind of a person.  At the
bottom of the lattice we have an element which corresponds to a the
inability of a machine or formal system to compute or define a value.
The transitions between the meanings we attach to specifications and
correct programs, and successive execution states of a correct program,
follow the information ordering, rather than the truth ordering.

\section{Related work}
\label{sec-related}

\citeN{Den_Bru_Mar:TOCL01} discuss
the use of inductive definition in mathematical logic.
They develop a general theory of induction over non-monotone
operators, and at the same time provide strong justification for
the well-founded semantics \cite{VanGelder91,Fitting:JLP93}
for logic programs with negation.
The view of \citeN{Den_Bru_Mar:TOCL01} is that recursive logic
programs represent inductive definitions --- the view to which we,
with this paper, also subscribe.
\citeN{Den_Bru_Mar:TOCL01} is not concerned with intended
semantics and specification, but the authors still make essential 
use of four-valued (as opposed to three-valued) logic, 
albeit primarily for reasons of technical convenience.

Arieli~\cite{Arieli:AMAI02} similarly gives a fixed point
characterisation of the meaning of logic programs.
One aim is to provide a language that supports knowledge revision
and reasoning with uncertainty.
Arieli's logic programming language has two kinds of negation,
namely explicit negation ($\neg$) and negation-by-failure
(\texttt{not}).
The proposed semantics allows for paraconsistency, that is, the
handling of locally inconsistent information in a way that does
not lead to the entire program being considered inconsistent.
That context naturally leads to the use of Belnap's logic.

\citeN{Loyer:TOCL04} are similarly concerned with an extended language.
In this case, the language is that of ``Fitting programs'', the kind
of logic programs used by \citeN{fitting:JLP_1991}, with the usual
connectives ``duplicated'' for the bilattice $\Quad$.
\citeN{Loyer:TOCL04} extend Fitting's work on reasoning in a
distributed (or multi-agent) setting.
The semantic framework they propose separates ``hypotheses'' from
``facts'' and is broad enough that, when restricted to Datalog
programs, it generalises both Fitting's ``Kripke-Kleene'' semantics
\cite{Fitting85} and the well-founded semantics \cite{VanGelder91}.
The framework, which again is based on four-valued logic, provides
what can be seen as a well-founded semantics for Fitting programs.

Many-valued logics have also long been advocated outside 
the logic programming community, but
the take-up there has arguably been more limited. 
In Section~\ref{sec-specification} we briefly mentioned the aims and
ideas of the Vienna Development Method (VDM).
This school has long argued that since programs, functions, and 
procedures that are written in a Turing complete language may be partial,
some sort of ``logic for partial functions'' is needed,
and that such a logic necessarily is three-valued.
As an extension,
Arieli and Avron
\citeyear{arieli-avron:1996,arieli-avron:LICS1998}
have argued the case for four-valued logic.

Starting with \citeN{McCarthy:FormalBasis63},
many have argued in favour of many-valued logics in which connectives
such as $\land$ and $\lor$ are no longer commutative.
For example, in McCarthy's logic, $\textbf{t} \lor \textbf{u}$ is 
equivalent to $\textbf{t}$, but
$\textbf{u} \lor \textbf{t}$ is equivalent to $\textbf{u}$ 
(whereas in $K_3$ it is $\textbf{t}$ as well).
The lack of commutativity makes these connectives implementable in 
a \emph{sequential} programming language,
and it corresponds closely to how the connectives are defined in
most modern programming languages.
As an example of the use of many-valued logic with non-commutative
conjunction, \citeN{barbuti:SCP1998} give a pure-Prolog semantics 
which is designed to closely mirror Prolog's depth-first left-to-right
evaluation strategy.
The logic has four truth values, and the roles of \textbf{u},
\textbf{f}, and \textbf{t} are conventional.
However, the fourth value, denoted $\textbf{t}_u$, is very different
to \textbf{i}.
Its operational-semantics reading is that it stands for
divergence preceded by success.
Another example of the use of non-commutative conjunction is
the three-valued logic proposed by \citeN{avron-konikowska:2009}
which combines $K_3$ (for reasoning about parallel constructs)
with McCarthy's logic (for reasoning about sequential constructs).

\citeN{morris-bunkenburg:1998} are concerned with 
\emph{program refinement}
in the presence of partiality and non-determinism 
(in program statements and/or in specifications).
They present a four-valued calculus over a language which includes
a (non-monotone) ``defined'' predicate, a device also used in LPF.

\citeN{Chechik-TOSEM:2003} use Belnap's 4-valued logic for the
analysis of so-called mixed transition systems~\cite{Dams-Toplas:1997}.
Transitions in mixed transition systems carry a modality (may or must)
with no assumption that a ``must'' transition also necessarily is a
``may'' transition.
As a consequence, it is possible for a property to both hold and not
hold.

\citeN{nishimura:2009} uses the simple 4-valued bilattice
$\Quad$ in a variant of refinement calculus.
Predicate transformers are developed for a small language of
program statements, including an exception catching primitive,
for use with abnormal program behaviour (division by zero, say)
as well as with explicit programmer-raised exceptions.
Nishimura's use of $\Quad$, however,
does not reflect a need to capture
varying degrees of information content.
Rather, four-valued logic is used to provide an elegant encoding trick.
In Nishimura's setting, $wp(S,\varphi_n, \varphi_e)$ expresses the
weakest condition which, when it holds before program statement $S$,
will ensure that, either, $S$ terminates normally, making $\varphi_n$ true,
or else $S$ terminates abnormally, making $\varphi_e$ true.
The four-valued lattice provides a convenient way of representing the
four possible states of the pair $\tuple{\varphi_n, \varphi_e}$.

The use of many-valued logic for reasoning about programs has also
had its detractors who argue that abandoning classical logic
complicates things, for insufficient gain.
\citeN{gries-schneider:1995} are concerned that
three-valued logics abandon the law of the excluded middle,
so that the schema $\varphi \lor \neg \varphi$ no longer is valid.
They point out that, if $\mathbf{u} \biim \mathbf{u}$ is valid
(and they insist that every instance of $\varphi \biim \varphi$
ought to be valid)
then the bi-implication connective $\biim$
fails to be associative, since we otherwise would have
\[
\textbf{f}~~ \equiv 
~~\textbf{t} \biim \textbf{f}~~ \equiv
~~(\textbf{u} \biim \textbf{u}) \biim \textbf{f}~~ \equiv
~~\textbf{u} \biim (\textbf{u} \biim \textbf{f})~~ \equiv 
~~\textbf{u} \biim \textbf{u}~~ \equiv 
~~\textbf{t}
\]
that is, we would have inconsistency.
They conclude that three-valued logic is too complicated to use
and favour staying instead with 2-valued logic by somehow side-stepping
non-denoting terms.
Problematic terms should be carefully prefixed to avoid non-denotation. 
For example,
``$y/y = 1$'' should systematically be replaced by
``$y \not= 0 \impl y/y = 1$''.

To us it seems that \citeN{gries-schneider:1995} ask for too much.
It is only to be expected that the law of the excluded middle will
be lost once we allow non-denoting terms in statements. 
And in the context of non-denoting terms, taking 
$\mathbf{u} \biim \mathbf{u}$ as valid would seem counter-intuitive.
It is a stretch to consider the statement 
$n/0 = 42~~ \biim ~~n/0 = 5$ valid,
given that $42 \not= 5$.
A far more natural approach is to consider that statement
ill-defined, that is, being neither true nor false.
As for the guarding of a formula $\varphi$ by conditions that ensure 
all terms in $\varphi$ are denoting,
that is hardly a practical solution when the terms involved stem 
from a Turing complete language---in place of ``$y/y = 1$'' 
consider being confronted with ``$f(y) = 1$'',
where $f$ has been given a (possibly complex) recursive definition.

\section{Conclusion}
\label{sec-conclusion}

Four-valued logic has previously been suggested as a tool for
reasoning about program behaviour in the context of
partiality, non-determinism and underspecification.
In a logic programming context, it has been used for parallel and
distributed programming, both as a language feature
\cite{fitting:JLP_1991} and an analysis tool \cite{Palmer97}.
In this paper we have argued that four-valued logic provides a handle
on many different situations that call for
reasoning about logic programs, 
even when we restrict attention to sequential programming.
The applications include
program analysis, 
type and mode systems,
formal specification, 
and declarative debugging.
Moreover, a semantics based on four truth values turns out
to be no more complex than one based on three.

Logicians have been aware of the limitations of formal systems since well
before the invention of electronic computers.  G{\"o}del showed the 
impossibility of a complete proof procedure for elementary number theory,
hence important gaps between truth and provability,
and in any Turing-complete programming
language there are programs which fail to terminate---undefinedness
is unavoidable.  Our awareness of the limitations of humans in
their interaction with computing systems goes back even further.
\citeN{babbage} claims to have been asked by members of the 
Parliament of the United Kingdom,
``Pray, Mr.\ Babbage, if you put into the machine wrong figures,
will the right answers come out''?  The term ``garbage in, garbage out''
was coined in the early days of electronic computing and concepts such as
``preconditions'' have always been important in formal verification of
software---underspecification is also unavoidable in practice.

Using a special value to denote undefinedness is the accepted practice
in programming language semantics.  Using a special value to denote
underspecification is less well established, but has been shown to
provide elegant and natural reasoning about partial correctness, at
least in the logic programming context.  In this paper we have proposed
a domain for reasoning about Prolog programs which has values to denote
both undefinedness and underspecification---they are the bottom and top
elements of a bilattice.  This gives an elegant picture which encompasses
both humans not making sense of some things and computers 
being unable to produce definitive results sometimes.  The logical connectives
Prolog uses in the body of clauses operate within the truth order in
the bilattice.  However, the overall view of computation does not
operate in the truth order, it operates
in the orthogonal ``information'' order.

\bibliographystyle{acmtrans}

\end{document}